\documentclass[11pt, letterpaper]{article}
\usepackage[margin=1in]{geometry}
\pdfoutput=1

\title{A Unified PTAS for Prize Collecting TSP and Steiner Tree Problem in Doubling Metrics}
\date{}

\author{T-H. Hubert Chan\thanks{Department of Computer Science, The University of Hong Kong. {\texttt{hubert@cs.hku.hk}}} \and Haotian Jiang\thanks{Department of Physics, Tsinghua University. {\texttt{jht14@mails.tsinghua.edu.cn}}} \and Shaofeng H.-C. Jiang\thanks{The Weizmann Institute of Science. {\texttt{shaofeng.jiang@weizmann.ac.il}}}}

\usepackage{amsmath}
\usepackage{amssymb}
\usepackage{amsthm}

\usepackage{url}
\usepackage{color}
\usepackage{paralist}
\usepackage{hyperref}
\usepackage{xspace}
\usepackage{boxedminipage}
\usepackage{graphicx}
\usepackage{float}

\newtheorem{fact}{Fact}[section]
\newtheorem{claim}{Claim}[section]
\newtheorem{remark}{Remark}[section]

\newtheorem{theorem}{Theorem}[section]
\newtheorem{lemma}{Lemma}[section]
\newtheorem{proposition}{Proposition}[section]
\newtheorem{corollary}{Corollary}[section]
\newtheorem{definition}{Definition}[section]

\newcommand{\eps}{\epsilon}

\newcommand{\enr}{\eps}

\newcommand{\Par}{\ensuremath{\mathsf{Par}}\xspace}
\newcommand{\Anc}{\ensuremath{\mathsf{Anc}}\xspace}

\newcommand{\R}{\mathbb{R}}

\newcommand{\poly}{\operatorname{poly}}

\newcommand{\TSP}{\ensuremath{\mathsf{TSP}}\xspace}

\newcommand{\PCTSP}{\ensuremath{\mathsf{PC^{TSP}}}\xspace}
\newcommand{\PCSTP}{\ensuremath{\mathsf{PC^{STP}}}\xspace}
\newcommand{\PCX}{\ensuremath{\mathsf{PC^X}}\xspace}
\newcommand{\X}{\ensuremath{\mathsf{X}}\xspace}

\newcommand{\STP}{\ensuremath{\mathsf{STP}}\xspace}
\newcommand{\INS}{\ensuremath{W}\xspace}

\newcommand{\OPT}{\ensuremath{\mathsf{OPT}}\xspace}
\newcommand{\OPTnr}{\ensuremath{\mathsf{OPT^{NR}}}\xspace}

\newcommand{\T}{\ensuremath{\mathsf{T}}\xspace}
\newcommand{\heur}{\ensuremath{\mathsf{H}}\xspace}
\newcommand{\sub}[1]{|_{#1}}

\newcommand{\Exp}{\ensuremath{\mathsf{Exp}}\xspace}
\newcommand{\ALG}{\ensuremath{\mathsf{ALG}}\xspace}
\newcommand{\DP}{\ensuremath{\mathsf{DP}}\xspace}
\newcommand{\Diam}{\ensuremath{\mathsf{Diam}}}
\newcommand{\expct}[1]{\ensuremath{\text{{\bf E}$\left[#1\right]$}}}
\newcommand{\wT}{\ensuremath{w^{\mathsf{T}}}\xspace}
\newcommand{\wP}{\ensuremath{w^{\mathsf{P}}}\xspace}

\newcommand{\ignore}[1]{}

\newcommand{\hubert}[1]{{\footnotesize\color{red}[Hubert: #1]}}
\newcommand{\shaofeng}[1]{{\footnotesize\color{blue}[Shaofeng: #1]}}
\newcommand{\haotian}[1]{{\footnotesize\color{blue}[Haotian: #1]}}

\begin{document}

\begin{titlepage}

\maketitle

\begin{abstract}
We present a unified (randomized) polynomial-time approximation scheme (PTAS) for the prize collecting traveling salesman problem (PCTSP) and the prize collecting Steiner tree problem (PCSTP) in doubling metrics. Given a metric space and a penalty function on a subset of points known as terminals, a solution is a subgraph on points in the metric space, whose cost is the weight of its edges plus the penalty due to terminals not covered by the subgraph.  Under our unified framework, the solution subgraph needs to be Eulerian for PCTSP, while it needs to be a tree for PCSTP.  Before our work, even a QPTAS for the problems in doubling metrics is not known.

Our unified PTAS is based on the previous dynamic programming frameworks proposed in [Talwar STOC 2004] and [Bartal, Gottlieb, Krauthgamer STOC 2012]. However, since it is unknown which part of the optimal cost is due to edge lengths and which part is due to penalties of uncovered terminals, we need to develop new techniques to apply previous divide-and-conquer strategies and sparse instance decompositions.

\end{abstract}

\thispagestyle{empty}
\end{titlepage}

\section{Introduction}
\label{section:intro}

We study prize collecting versions of two important optimization problems: the prize collecting traveling salesman problem ($\PCTSP$) and the prize collecting Steiner tree problem ($\PCSTP$).
In both problems, we are given a metric space and a set of points called \emph{terminals}, and a non-negative penalty function on the terminals.
A solution for either problem is a connected subgraph with vertex set from the metric.
In addition, it needs to be an Eulerian (multi-)graph\footnote{An undirected connected multi-graph is Eulerian, if every vertex has even degree.}for $\PCTSP$ and a tree for $\PCSTP$.
The cost of a solution is the sum of the weights of edges in the solution plus the sum of penalties due to terminals not visited by the solution. 

\noindent\textbf{Prize Collecting Problems in General Metrics.}
The prize collecting setting was first considered by Balas~\cite{balas1989prize}, who proposed the prize collecting TSP.
However, the version that Balas considered is actually more general,
in the sense that each terminal is also associated with a reward, and the goal is to find a tour that minimizes the tour length plus the penalties, and collects at least a certain amount of rewards.
The setting that we consider was suggested by Bienstock et al.~\cite{bienstock1993note}, and they
used LP rounding to
 give a 2.5-approximation algorithm for the $\PCTSP$ and a 3-approximation for the $\PCSTP$.
Later on, a unified primal-dual approach for several network design problems was proposed~\cite{DBLP:journals/siamcomp/GoemansW95};
this approach improves the approximation ratios for both $\PCTSP$ and $\PCSTP$ to 2 in general metrics.
The 2-approximation had remained the state of the art for more than a decade, until Archer et al.~\cite{DBLP:journals/siamcomp/ArcherBHK11} finally broke the 2 barrier for both problems.
Subsequently, in a note~\cite{goemans2009combining}, Goemans combined their argument with other algorithms, and gave a 1.915-approximation for the $\PCTSP$, which is the state of the art.

\noindent\textbf{Prize Collecting Problems in Bounded Dimensional Euclidean Spaces.}
$\PCTSP$ and $\PCSTP$ are APX-hard in general metrics, because even the special cases, the TSP and the Steiner tree problem, are APX-hard.
Although the seminal result by Arora~\cite{DBLP:journals/jacm/Arora98} showed that both $\TSP$ and $\STP$ have PTAS's in bounded dimensional Euclidean spaces, the prize collecting setting was not discussed.
However, we do believe that their approach may be directly applied to get PTAS's for the prize collecting versions of both problems, with a slight modification to the dynamic programming algorithms.
Later, A PTAS for the Steiner Forest Problem (which generalizes the $\STP$) was discovered by Borradaile et al.~\cite{DBLP:journals/talg/BorradaileKM15}.
Based on this result, Bateni et al.~\cite{DBLP:journals/algorithmica/BateniH12} studied the Prize Collecting Steiner Forest Problem, and gave a PTAS for the special case when the penalties are multiplicative, but this does not readily imply a PTAS for the $\PCTSP$ or the $\PCSTP$.

\noindent\textbf{Prize Collecting Problems in Special Graphs.}
Planar graphs is an important class of graphs. Both problems are considered in planar graphs, and a PTAS is presented by Bateni et al.~\cite{bateni2011prize} for $\PCTSP$ and $\PCSTP$. Moreover, they noted that both problems are solvable in polynomial time in bounded treewidth graphs, and their PTAS relies on a reduction to the bounded treewidth cases. They also showed that the Prize Collecting Steiner Forest Problem, which is a generalization of the $\PCSTP$, is significantly harder, and it is APX-hard in planar graphs and Euclidean instances. As for the minor forbidden graphs, which generalizes planar graphs, there are PTAS's for various optimization problems, such as TSP by Demaine et al.~\cite{demaine2011contraction}. However, the PTAS's for prize collecting problems, to the best of our knowledge, are unknown.

\ignore{
Planar graphs is an important class of graphs. Both problems are considered in planar graphs, and a PTAS is presented by Bateni et al.~\cite{bateni2011prize} for $\PCTSP$ and $\PCSTP$. Moreover, they noted that both problems are solvable in polynomial time in bounded treewidth graphs, and their PTAS replies on a reduction to the bounded treewidth cases. They also showed that the Prize Collecting Steiner Forest Problem, which is a generalization of the $\PCSTP$, is significantly harder, and it is APX-hard in planar graphs and Euclidean instances. As for the minor forbidden graphs, which generalizes planar graphs, although there are PTAS's for various optimization problems, such as TSP by Demaine et al.~\cite{demaine2011contraction}, the PTAS's for prize collecting problems are unknown. \shaofeng{Haotian, please confirm this.}
}

\noindent\textbf{Generalizing Euclidean Dimension.}
Going beyond Euclidean spaces, doubling dimension~\cite{Assouad83,DBLP:journals/dcg/Clarkson99,DBLP:conf/focs/GuptaKL03} is a popular notion of dimensionality. It captures the bounded local growth of Euclidean spaces, and does not require any specific Euclidean properties such as vector representation or dot product.
A metric space has doubling dimension at most $k$, if every ball can be covered by at most $2^k$ balls of half the radius. This notion generalizes the Euclidean dimension, in that every subset of $\mathbb{R}^d$ equipped with $\ell_2$ has doubling dimension $O(d)$. Although doubling metrics are more general than Euclidean spaces, recent results show that many optimization problems have similar approximation guarantees for both spaces: there exist PTAS's for the TSP~\cite{DBLP:journals/siamcomp/BartalGK16}, a certain version of the TSP with neighborhoods~\cite{DBLP:conf/soda/ChanJ16}, and the Steiner forest problem~\cite{DBLP:conf/focs/ChanHJ16}, in doubling metrics.

\noindent\textbf{Our Contributions.} In this paper, we extend this line of research, and give a unified PTAS framework for both $\PCTSP$ and $\PCSTP$. We use $\PCX$ when the description applies to either problem.
Our main result is Theorem~\ref{thm:main}.

\begin{theorem}
\label{thm:main}
For any $0<\epsilon<1$, there exists an algorithm that,
for any $\PCX$ instance with~$n$ terminal points in a metric space with doubling dimension at most~$k$, runs in time
\begin{equation*}
	n^{O(1)^{O(k)}} \cdot \exp( \sqrt{\log{n}} \cdot O(\frac{k}{\epsilon})^{O(k)} ),
\end{equation*}
and returns a solution that is a $(1+\epsilon)$-approximation with constant probability.
\end{theorem}

\noindent\textbf{Technical Issues.}
As a first trial, one might try to adapt the sparsity framework used in previous PTAS's for the TSP and Steiner forest problems~\cite{DBLP:journals/siamcomp/BartalGK16,DBLP:conf/soda/ChanJ16,DBLP:conf/focs/ChanHJ16} in doubling metrics.
The framework typically uses a polynomial-time estimator $\heur$ on any ball $B$,
which gives a constant approximation for \PCX on some appropriately defined sub-instance around~$B$.
Intuitively, the estimator works because the local behavior of a (nearly) optimal solution can be well estimated by looking at the sub-instance locally.
In particular, the following properties are needed in this framework:

\begin{compactitem}
	\item If $\heur(B)$ is large, then the optimal solution for the sub-instance induced on $B$ is large; moreover, any (nearly) optimal solution for the global instance would have a large part of its cost due to~$B$.
	\item If $\heur(B)$ is small, then for any (nearly) optimal solution~$F$ for the global instance,
	the cost of~$F$ contributed by the sub-instance due to~$B$ should be small.
\end{compactitem}

While the first property is somehow straightforward,
the following example shows that the second property is non-trivial to achieve in \PCX.

\noindent\textbf{Example Instance: Figure~\ref{fig:hard}.}
The example is defined on the real line.  The terminals are grouped into two clusters.
The left cluster contains $2m$ terminals, and the right cluster contains $m$ terminals.
Within each cluster, the distance between adjacent terminals is 1.  The two clusters
are at distance $l$ apart.  The penalty for each terminal is~$t$.
The parameters are chosen such that $l \gg mt$ and  $t \gg m$.  Observe that for \PCX,
the optimal solution is to visit all the terminals in the left cluster with total edge weights $O(m)$ and incur the penalty $m t$ for the terminals in the right cluster.
The reason is that it will be too costly to add an edge to connect terminals
from different clusters, and it is better to visit the cluster with more terminals and
suffer the penalty for the cluster with fewer terminals.


\begin{figure}[H]
	\centering
	\includegraphics[width=0.8\textwidth]{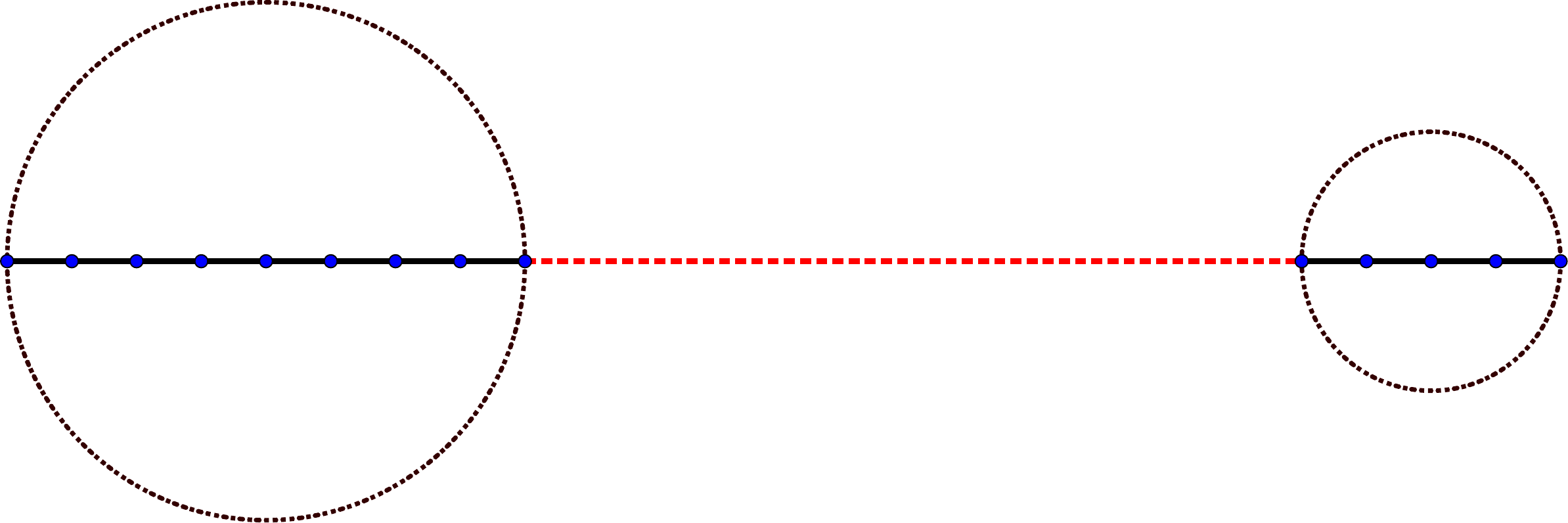}
	\caption{Example instance for \PCX}
	\label{fig:hard}
\end{figure}


\ignore{

The example is defined in the Euclidean plane, and all points have integer coordinates.
The points are defined as grids in each of the two squares that are both of edge length $m$, and the two squares are of distance at least $l$.
The penalty for each point is $t$.
Let $l \gg t \gg m^2$. Note that there are $(m+1)^2$ grid points for each square.
An optimal solution would connect all grids inside exactly one square, with cost $O(m^2)$, and does not connect any other points at all.
One may wonder why we do not connect both squares. The reason is that the solution has to be a tour, and connecting both squares requires an edge of length at least $l$, which is too costly.
}

\noindent\textbf{Local Estimator Fails on the Example Instance.}
Suppose the estimator is applied around a ball~$B$ centered at some terminal in the right cluster with radius $r$. Then, any constant-approximate solution for the sub-instance needs to connect all $\Theta(r)$ terminals in the ball, since the penalty for any single terminal is too large.
This costs $\Theta(r)$. However, in the optimal solution, no terminal in the right cluster is visited and all penalties are taken, which has cost $\Omega(t r)$. Hence, the estimator fails to serve as an upper bound for the contribution by ball~$B$ to the cost of an optimal solution.

The conclusion is that the optimal solution of a local sub-instance can differ a lot from
how an optimal global solution behaves for that sub-instance.

\noindent\textbf{Our Insight: Trading between Weight and Penalty.} Our example in
Figure~\ref{fig:hard} shows that what points a local optimal solution visits
in a sub-instance can be very different from the points in the sub-instance
visited by a global optimal solution.  Our intuition is that the optimal cost
of a sub-instance should reflect part of the cost in a global optimal solution due
to the sub-instance.  In other words, if a sub-instance has large optimal cost,
then any global solution either (1) has a large weight within the sub-instance,
or (2) suffers a large penalty due to unvisited terminals in the sub-instance.  This insight leads to the following key ingredients to our solution.

\begin{compactitem}

\item[1.] \textbf{Inferring Local Behavior from Estimator.}  In Lemma~\ref{lemma:sparsity_estimator}, we show that the value returned by the local estimator (which consists
of both the weight and the penalty) on a ball~$B$ gives an upper bound
on the weight $w(F|_B)$ of any (near) optimal solution~$F$ inside ball~$B$.
We emphasize that this estimator is an upper bound for the \emph{weight} $w(F|_B)$ \emph{only}, and is \emph{not} an upper bound for both the weight and penalty of the optimal solution inside the ball.
In the example in Figure~\ref{fig:hard}, a global optimal solution does not visit
the right cluster at all, and hence, the local estimator on the right cluster does give an upper bound on the \emph{weight} part of the global solution due to the right cluster.
This turns out to be sufficient because the sparsity of a solution is defined with respect to only the
weight part (and not the penalty part).

Hence, the local estimator can be used in the sparsity decomposition framework~\cite{DBLP:journals/siamcomp/BartalGK16,DBLP:conf/soda/ChanJ16,DBLP:conf/focs/ChanHJ16}
to identify a \emph{critical} instance $\INS_1$ (i.e., the local estimator reaches some threshold, but still not too large) around some ball~$B$.
Since the instance~$\INS_1$ is sparse enough, an approximate solution~$F_1$ can be obtained by the dynamic program framework.
Then, one can recursively solve for an approximate solution~$F_2$ for the remaining instance~$\INS_2$.  However, we
need to carefully define $\INS_2$ and combine the solutions $F_1$ and $F_2$,
because, as we remarked before, even if the approximate algorithm returns~$F_1$ for the instance~$\INS_1$,
a near optimal global solution might not visit any terminals in~$\INS_1$.

\item[2.] \textbf{Adaptive Recursion.} In all previous applications
of the sparsity decomposition framework, after a critical ball~$B$ around some center~$u$ is
identified, the original instance is decomposed into sub-instances $\INS_1$ and $\INS_2$ that
can be solved independently.

An issue in applying this framework is that after obtaining solutions~$F_1$ and $F_2$ for the sub-instances,
in the case that $F_1$ and $F_2$ are far away from each other as in our example in Figure~\ref{fig:hard} where it is too costly to connect them directly,
it is not clear immediately which of $F_1$ and $F_2$ should be the weight part of the global solution and which would become the penalty part.

We use a novel idea of the adaptive recursion, in which $\INS_2$ depends on the solution~$F_1$ returned for~$\INS_1$.
The high level idea is that in defining the instance~$\INS_2$,
we add an extra terminal point at~$u$, which becomes a representative for solution~$F_1$.
The penalty of~$u$ in~$\INS_2$ is the sum of the penalties of terminals in~$\INS_1$ minus
the cost $c(F_1)$ of solution~$F_1$.  After a solution~$F_2$ for $\INS_2$ is returned,
if $F_2$ does not visit the terminal~$u$, then edges in~$F_1$ are discarded,
otherwise the edges in~$F_1$ and $F_2$ are combined to return a global solution. \ignore{\haotian{Shall we add some intuition here, for example, saying that the point $u$ serves as an incentive for the algorithm to connect~$F_2$ to~$F_1$.}}

We can see that in either case, the sum $c(F_1) + c(F_2)$ of the costs
of the two solutions reflect the cost of the global solution.
In the first case, $F_2$ does not visit~$u$ and hence, $c(F_2)$
contains the penalty due to $u$, which is the penalties of unvisited terminals
in $\INS_1$ minus $c(F_1)$. Therefore, when $c(F_1)$ is added back,
the sum simply contains the original penalties of unvisited terminals in~$\INS_1$.

In the second case, $F_2$ does visit~$u$ and does not incur a penalty due to~$u$.
Therefore, $c(F_1) + c(F_2)$ does reflect the cost of the global solution after combining $F_1$ and $F_2$.
\end{compactitem}

\noindent\textbf{Revisiting the Sparsity Structural Lemma.}
Many PTAS's in the literature for TSP-like problems in doubling metrics
rely on the sparsity structural lemma~\cite[Lemma 3.1]{DBLP:journals/siamcomp/BartalGK16}.
Intuitively, it says that if a solution is sparse, then there exists
a structurally light solution that is $(1+\epsilon)$-approximate. Hence, one can restrict
the search space to structurally light solutions, which can be explored by a dynamic program algorithm.  Because of the significance of this lemma, we believe that it is worthwhile
to give it a more formal inspection, and in particular, resolve some significant technical issues
as follows.
\begin{compactitem}
	\item \emph{Issue with Conditioning on the Randomness of Hierarchical Decomposition.}
	Given a hierarchical decomposition and a solution~$T$, the first step is to reroute the solution
	such that every cluster is only visited through some designated points known as \emph{portals}.
	The randomness in the hierarchical decomposition is used to argue that
	the expected increase in cost to make the solution portal-respecting is small.
	
	However, typically the randomness in the hierarchical decomposition is still needed in subsequent arguments.  Hence, if one analyzes the portal-respecting procedure as a conceptually separate step, then subsequent uses of the randomness of the hierarchical decomposition
	need to condition on the event that the portal-respecting step does not increase the cost too much.  Moreover, edges added in the portal-respecting step are actually random objects
	depending on the hierarchical decomposition, and hence, will in fact cross some clusters with probability 1. Unfortunately, even in the original paper by Talwar~\cite{DBLP:conf/stoc/Talwar04} on the QPTAS for TSP in doubling metrics, these issues were not addressed properly.	
	\item \emph{Issues with Patching Procedure.}  A patching procedure is typically used to reduce
	the number of times a cluster is crossed. In the literature, after reducing the number of crossings, the triangle inequality is used to implicitly add some shortcutting edges outside the cluster.  However, it is never argued whether these new shortcutting edges are still portal-respecting.  It is plausible that making them portal-respecting might introduce new crossings.
\end{compactitem}

From the above discussion, it is evident that one should consider the portal-respecting step and the patching procedure together,
because they both rely on the randomness of the hierarchical decomposition.
To make our arguments formal, we need a more precise notation to describe portals, and in Section~\ref{section:ptas_sparse}, we actually revisit the whole randomized
hierarchical decomposition to make all relevant definitions precise.
In Theorem~\ref{theorem:struct}, we analyze the portal-respecting step and the patching procedure together
through a sophisticated accounting argument so that the patching cost is eventually charged back to the original solution
(as opposed to stopping at the transformed portal-respecting solution).

Moreover, we give a unified patching lemma that works for both $\PCTSP$ and $\PCSTP$.
Even though our proofs use similar ideas as previous works, the charging argument is significantly different.
Specifically,
our argument does not rely on the small MST lemma~\cite[Lemma 6]{DBLP:conf/stoc/Talwar04}, which was also used in~\cite{DBLP:journals/siamcomp/BartalGK16}.

\noindent\textbf{Paper Organization.}
Section~\ref{section:prelim} gives the formal notation and describes the outline of the sparsity decomposition framework to solve \PCX.
Section~\ref{section:estimator} gives the properties of the local sparsity estimator.
Section~\ref{section:decomposition} gives the technical details of the sparsity decomposition
and shows that approximate solutions in sub-instances can be combined to give a good approximation
to the global instance.  Section~\ref{section:ptas_sparse} revisits the hierarchical decomposition
and sparse instance frameworks for TSP-like problems in doubling metrics;
the notation is more involved than previous works, and readers who are already familiar
with the literature might choose to skip it during the first read.
Section~\ref{section:dp} gives the details of the dynamic program for sparse instances and
the analysis of its running time.

\section{Preliminaries}
\label{section:prelim}

\ignore{
\shaofeng{Only the ``problem definition'' and ``rescaling instance'' is rewritten somehow; other parts before Section 2.1 is untouched.}
}
We consider a metric space $M=(X,d)$
(see~\cite{DL97,mat-book} for more details on metric spaces),
where we refer to an element $x \in X$ as a point or a vertex.
For $x \in X$ and $\rho \geq 0$, a \emph{ball} $B(x,\rho)$ is the set $\{y \in X \mid d(x,y) \leq \rho\}$.
The \emph{diameter} $\Diam(Z)$ of a set $Z \subset X$ is the maximum distance between points in $Z$.
For $S, T \subset X$,
we denote $d(S,T) := \min\{d(x,y): x \in S, y \in T\}$,
and for $u \in X$, $d(u,T) := d(\{u\}, T)$.
Given a positive integer $m$, we denote $[m] := \{1,2,\ldots, m\}$.

A set $S \subseteq X$ is a $\rho$-packing, if any two distinct
points in $S$ are at a distance more than $\rho$ away from each other.
A set $S$ is a $\rho$-cover for $Z \subseteq X$,
if for any $z \in Z$, there exists $x \in S$ such that $d(x,z) \leq \rho$.
A set $S$ is a $\rho$-net for $Z$, if $S$ is a $\rho$-packing
and a $\rho$-cover for $Z$.
We assume the access to an oracle that takes a series of balls $\{B_i\}_i$ where each $B_i$ is identified by the center and radius, and returns a point $x\in X$ such that $\forall i, x \notin S_i$\footnote{Such an oracle is trivial to construct for finite metric spaces. It may also be efficiently constructed for many special infinite metric spaces, such as bounded dimensional Euclidean spaces.}. A greedy algorithm can construct a $\rho$-net efficiently given the access to this oracle.


We consider metric spaces with \emph{doubling dimension}~\cite{Assouad83,DBLP:conf/focs/GuptaKL03} at most $k$;
this means that
for all $x \in X$, for all $\rho> 0$, every ball $B(x,2\rho)$ can be covered
by the union of at most $2^k$ balls of the form $B(z, \rho)$, where $z
\in X$.  The following fact captures a standard property
of doubling metrics.


\begin{fact}[Packing in Doubling Metrics~\cite{DBLP:conf/focs/GuptaKL03}] \label{fact:net}
	Suppose in a metric space with doubling dimension at most $k$, a $\rho$-packing $S$ has diameter at most $R$.
	Then, $|S| \leq (\frac{2R}{\rho})^k$.
\end{fact}

\noindent\textbf{Edges.} An edge\footnote{To have a complete description, we also need the notion of self-loop, which is simply a singleton $\{x\}$.} $e$ is an unordered pair $e = \{x,y\} \in {X \choose 2}$
whose weight $w(e) = d(x,y)$ is induced by the metric space $(X,d)$.
Given a set $F$ of edges,
its vertex set $V(F) := \cup_{e \in F} \, e \subset X$
\ignore{\haotian{@Shaofeng: The `e' in the middle should be $v(e)$ or something like that?}}
is the vertices covered (or \emph{visited}) by the edges in $F$.
If $T \subset X$ is a set of vertices,
we use the shorthand $T \setminus F := T \setminus V(F)$ to
denote the vertices in $T$ that are not covered by $F$.


\noindent\textbf{Problem Definition.}
We give a unifying framework for the prize collecting traveling salesman problem ($\PCTSP$) and the prize collecting Steiner tree problem ($\PCSTP$)
\ignore{\haotian{@Shaofeng: I feel that this looks too similar to \PCTSP that can cause confusion, shall we write it simply as PCST?}}, and we use \PCX when the description
applies to both problems.  An instance $\INS = (T,\pi)$ of $\PCX$ consists of a set $T \subset X$ of \emph{terminals} (where $|W| := |T| = n$) and a penalty function $\pi: T \to \R_+$.
The goal is to find a (multi-)set $F \subset {X \choose 2}$ of edges
with minimum cost\footnote{When the context is clear, we drop the subscript in $c_\INS(\cdot)$.} $c_\INS(F) := w(F) + \pi(T \setminus F)$, such that
the following additional conditions are satisfied for each specific problem:
\begin{compactitem}
	\item For $\PCTSP$, the edges in the multi-set $F$ form a circuit on $V(F)$;
	for $|V(F)| = 1$, $F$ contains only a single self-loop (with zero weight).
	
	\item For $\PCSTP$, the edges $F$ form a connected graph on $V(F)$, where we also
	allow the degenerate case when $F$ is a singleton containing a self-loop.
	The vertices in $V(F) \setminus T$ are known as \emph{Steiner} points.
\end{compactitem}



\noindent\textbf{Simplifying Assumptions and Rescaling Instance.} Fix some constant $\eps > 0$.
Since we consider asymptotic running time to obtain $(1 + \eps)$-approximation for \PCX,
we consider sufficiently large $n > \frac{1}{\eps}$.
Since $F$ can contain a self-loop, an optimal solution covers at least
one terminal $u$.  Moreover, there is some terminal $v$ (which could be the same as $u$)
such that the solution covers $v$ and does not cover any terminal $v'$ with $d(u,v') > d(u,v)$.
Since we aim for polynomial time algorithms, we can afford to enumerate the $O(n^2)$ choices for $u$ and $v$.

For some choice of $u$ and $v$, suppose $R := d(u, v)$.
Then, $R$ is a lower bound on the cost of an optimal solution.
Moreover, the optimal solution $F$ has weight $w(F)$ at most $nR$, and hence, we do not need to consider points at distances larger than $nR$ from $u$
\ignore{\haotian{@Shaofeng: if we already fix $u$ and $v$, this means that the tour will be in $B_u(R)$, so there is no edge of length $nR$ right?}}.
Since $F$ contains at most $2n$ edges (because of Steiner points in $\PCSTP$), if we consider an $\frac{\eps R}{32n^2}$-net $S$ for $X$ and replace every point in $F$ with its closest net-point in $S$, the cost increases by at most $\eps \cdot \OPT$.
Hence, after rescaling, we can assume that inter-point distance is at least 1 and
we consider distances up to $O(\frac{n^3}{\eps}) = \poly(n)$.
By the packing property of doubling dimension (Fact~\ref{fact:net}),
we can hence assume $|X| \leq O(\frac{n}{\eps})^{O(k)} \leq O(n)^{O(k)}$.

\noindent\textbf{Hierarchical Nets.}
As in~\cite{DBLP:journals/siamcomp/BartalGK16},
we consider some parameter $s = (\log n)^{\frac{c}{k}} \geq 4$,
where $0<c<1$ is a universal constant that is sufficiently small (as required in Lemma~\ref{lemma:running_time}).
Set $L := O(\log_s n) = O(\frac{k \log n}{\log \log n})$.
A greedy algorithm can construct $N_L \subseteq N_{L-1} \subseteq \cdots \subseteq N_1 \subseteq N_0 = N_{-1}= \cdots = X$
such that for each $i$, $N_i$ is an $s^i$-net for $N_{i-1}$,
where we say \emph{distance scale} $s^i$ is of \emph{height} $i$.

\noindent\textbf{Net-Respecting Solution.}
As defined in~\cite{DBLP:journals/siamcomp/BartalGK16},
a graph $F$ is net-respecting
with respect to $\{N_i\}_{i\in[L]}$ and $\enr > 0$ if for
every edge $\{x,y\}$ in $F$, both
$x$ and $y$ belong to $N_i$, where $s^i \leq \enr \cdot d(x,y) < s ^{i+1}$.
By~\cite[Lemma 1.6]{DBLP:journals/siamcomp/BartalGK16}, any graph $F$ may be converted to a net-respecting $F'$ visiting all points that $F$ visits, and $w(F') \leq (1 + O(\epsilon))\cdot w(F)$.

Given an instance $\INS$ of a problem,
let $\OPT(\INS)$ be an optimal solution;
when the context is clear, we also use $\OPT(\INS)$ to denote
the cost $c(\OPT(\INS))$, which includes both its weight and the incurred penalty;
similarly, $\OPT^{nr}(\INS)$ refers to an optimal net-respecting solution.

\subsection{Overview}
\label{sec:sparse_overview}

We achieve a PTAS for $\PCX$ by a unified framework,
which is based on the framework of sparse instance decomposition as in~\cite{DBLP:journals/siamcomp/BartalGK16,DBLP:conf/soda/ChanJ16,DBLP:conf/focs/ChanHJ16}.

\noindent\textbf{Sparse Solution~\cite{DBLP:journals/siamcomp/BartalGK16}.} Given an edge set $F$ and  a subset $S \subseteq X$, $F|_S := \{ e \in F: e \subseteq S\}$ is the
edges in $F$ totally contained in $S$.
An edge set $F$ is called $q$-sparse, if
for all $i \in [L]$ and all $u \in N_i$,
$w(F|_{B(u, 3 s^i)}) \leq q \cdot s^i$.

\noindent\textbf{Sparsity Structural Property. (Revisited in Theorem~\ref{theorem:struct})}
An important technical lemma~\cite[Lemma 3.1]{DBLP:journals/siamcomp/BartalGK16} in this framework
states that if a (net-respecting) solution $F$ is sparse, then with constant probability, there is some $(1 + \eps)$-approximate solution $\widehat{F}$ that is \emph{structurally light} with respect
to some randomized \emph{hierarchical decomposition} (see Section~\ref{sec:hier_decomp}).  Then, a bottom-up dynamic program (given
in Section~\ref{section:dp}) based
on the hierarchical decomposition searches for the best solution with the lightness structural property in polynomial time.


\begin{remark}
\label{remark:why_fix}
We observe that this technical lemma is used crucially in all previous works
on PTAS's for TSP variants in doubling metrics.  Hence, we believe that its proof
should be verified rigorously. In Section~\ref{section:intro}, we outlined
the technical issue in the original proof~\cite{DBLP:journals/siamcomp/BartalGK16},
and this issue actually appeared as far as in the first paper on TSP for doubling metrics~\cite{DBLP:conf/stoc/Talwar04}.  In Section~\ref{section:ptas_sparse}, we give a detailed description
to complete the proof of this important lemma.
\end{remark}



\ignore{
\shaofeng{The following paragraph is unchanged.}
}

\noindent\textbf{Sparsity Heuristic.}
As in~\cite{DBLP:journals/siamcomp/BartalGK16,DBLP:conf/soda/ChanJ16,DBLP:conf/focs/ChanHJ16},
we estimate the local sparsity of an optimal net-respecting solution with a heuristic.
For $i \in [L]$ and $u \in N_i$, given an instance $\INS$,
the heuristic $\heur^{(i)}_u(\INS)$ is supposed to
estimate the sparsity of an optimal net-respecting solution
in the ball $B':=B(u, O(s^i))$.
We shall see in Section~\ref{section:estimator} that
the heuristic actually gives a constant approximation
to some appropriately defined sub-instance $\INS'$
in the ball $B'$.


\noindent\textbf{Divide and Conquer.} Once we have a sparsity estimator,
the original instance can be decomposed into sparse sub-instances, whose
approximate solutions can be found efficiently.  As we shall see, the partial solutions
are combined with the following extension operator.  The algorithm outline
is described in Figure~\ref{fig:alg}.

\begin{definition}[Solution Extension]
\label{defn:extension}
Given two partial solutions $F$ and $F'$ of edges,
we define the \emph{extension} of $F$ with $F'$ at point $u$ as
$F \looparrowleft_u F' :=
\begin{cases}
F \cup F', & \mathrm{if } \, u \in V(F) \cap V(F');\\
F, &\mathrm{otherwise.}
\end{cases}
$
\end{definition}
%
%

\ignore{

\shaofeng{The description of the generic algorithm is unchanged except for step 4, and the original step 5 is deleted and merged to step 4.}
}

\begin{figure}[h]
\begin{boxedminipage}{\textwidth}
\noindent \textbf{Generic Algorithm.}
We describe a generic framework that applies to $\PCX$.
Similar framework is also used in \cite{DBLP:journals/siamcomp/BartalGK16,DBLP:conf/soda/ChanJ16,DBLP:conf/focs/ChanHJ16}
to obtain PTAS's for TSP related problems.
Given an instance $\INS$, we describe the recursive
algorithm  $\ALG(\INS)$ as follows.
This description is mostly the same with that in~\cite{DBLP:conf/focs/ChanHJ16}, except that the decomposition in Step 4 is more involved.

\begin{compactitem}
	\item[1.] \textbf{Base Case.} If $|\INS|=n$ is smaller than some constant threshold,
	solve the problem by brute force, recalling that $|X| \leq O(\frac{n}{\eps})^{O(k)}$.
	
	\item[2.] \textbf{Sparse Instance.} If for all $i \in [L]$,
	for all $u \in N_i$, $\heur^{(i)}_u(\INS)$ is at most $q_0 \cdot s^i$, for some appropriate threshold $q_0$,
	call the subroutine
	$\DP(\INS)$ to return a solution, and terminate.
	
	\item[3.] \textbf{Identify Critical Instance.} Otherwise, let $i$ be the smallest height such that
	there exists $u \in N_i$ with \emph{critical} $\heur^{(i)}_u(\INS) > q_0 \cdot s^i$; in this case,
	choose $u \in N_i$ such that $\heur^{(i)}_u(\INS)$ is maximized.
	
	\item[4.] \textbf{Divide and Conquer.}
	Define a sub-instance $\INS_1$ from around the critical instance (possibly using randomness).
	Loosely speaking, $\INS_1$ is a sparse enough sub-instance induced in the region around
	$u$ at distance scale $s^i$. Since it is sparse enough, we apply the dynamic programming algorithm
	on $\INS_1$ and get solution $F_1$.
	
	We define an appropriate sub-instance $\INS_2$ \emph{with the information of} $F_1$.
	Intuitively, $\INS_2$ captures the remaining sub-problem not included in $W_1$.
	We emphasize that as opposed to previous work~\cite{DBLP:journals/siamcomp/BartalGK16,DBLP:conf/soda/ChanJ16,DBLP:conf/focs/ChanHJ16}, $\INS_2$ can depend on $F_1$ (through the choice of the
	penalty function).  Moreover, we ensure that any solution $F_2$ of $\INS_2$
	can be extended to $F_2 \looparrowleft_u F_1$ as a solution for $W$, and the following holds:
	\begin{equation}
		c_W(F_2 \looparrowleft_u F_1) \leq c_{W_1}(F_1) + c_{W_2}(F_2).
		\label{eq:directed_union}
	\end{equation}

	We solve $\INS_2$ recursively and suppose the solution is $F_2$.
	We note that $\heur^{(i)}_u(\INS_2) \leq q_0 \cdot s^i$, and hence the recursion will terminate.
	
	Moreover,
	the following property holds:
	
	\vspace{-10pt}
	
	\begin{equation}
	\expct{\OPT(\INS_1)}
	\leq \frac{1}{1 - \eps} \cdot (\OPT^{nr}(\INS) - \expct{\OPT^{nr}(\INS_2)}),
	\label{eq:decompose}
	\end{equation}
	
	where the expectation is over the randomness of the decomposition.
	
	We return $F := F_2 \looparrowleft_u F_1$ as a solution to $\INS$.
	\end{compactitem}
\end{boxedminipage}
\caption{Algorithm Outline} \label{fig:alg}
\end{figure}

	\ignore{
	\item[4.] \textbf{Decomposition into Sparse Instances.} Decompose the instance $\INS$ into appropriate
	sub-instances $\INS_1$ and $\INS_2$ (possibly using randomness).
	Loosely speaking, $\INS_1$ is a sparse enough sub-instance induced in the region around
	$u$ at distance scale $s^i$, and $\INS_2$ captures the rest.
	We note that $\heur^{(i)}_u(\INS_2) \leq q_0 \cdot s^i$ such that the recursion will terminate.
	The union of the solutions
	to the sub-instances will be a solution to $\INS$.  Moreover,
	the following property holds.
	
	\vspace{-10pt}
	
	\begin{equation}
	\expct{\OPT(\INS_1)}
	\leq \frac{1}{1 - \eps} \cdot (\OPT^{nr}(\INS) - \expct{\OPT^{nr}(\INS_2)}),
	\label{eq:decompose}
	\end{equation}
	
	where the expectation is over the randomness of the decomposition.

	\item[5.] \textbf{Recursion.} Call the subroutine $F_1 := \DP(\INS_1)$,
	and solve $F_2 := \ALG(\INS_2)$ recursively;
	return the union $F_1 \cup F_2$.
}

\noindent\textbf{Analysis of Approximation Ratio.}
We follow the inductive proof as in~\cite{DBLP:journals/siamcomp/BartalGK16}
to show that with constant probability (where
the randomness comes from $\DP$), $\ALG(\INS)$ in Figure~\ref{fig:alg} returns a solution
with expected length at most $\frac{1+\eps}{1-\eps} \cdot \OPT^{nr}(\INS)$,
where expectation is over the randomness of decomposition into sparse instances
in Step 4.

As we shall see,
in $\ALG(\INS)$, the subroutine $\DP$ is called
at most $\poly(n)$ times (either explicitly in the recursion or
in the heuristic $\heur^{(i)}$).  Hence, with constant probability,
all solutions returned by all instances of $\DP$ have appropriate approximation
guarantees.

Suppose $F_1$ and $F_2$ are solutions returned by $\DP(\INS_1)$ and
$\ALG(\INS_2)$, respectively. We use $c_{i}$ as a shorthand for $c_{W_i}$, for $i = 1, 2$, and $c$ as a shorthand for $c_W$.
Since we assume that $\INS_1$ is sparse enough
and $\DP$ behaves correctly, $c_{1}(F_1) \leq
(1+\eps) \cdot \OPT(\INS_1)$.
The induction hypothesis states that
$\expct{c_2(F_2) | \INS_2} \leq \frac{1+\eps}{1-\eps} \cdot \OPT^{nr}(\INS_2)$.

In Step 4, equation~(\ref{eq:decompose}) guarantees that
$\expct{\OPT(\INS_1)}
\leq \frac{1}{1 - \eps} \cdot (\OPT^{nr}(\INS) - \expct{\OPT^{nr}(\INS_2)})$.
By equation~(\ref{eq:directed_union}), $c(F_2\looparrowleft_u F_1) \leq c_1(F_1) + c_2(F_2)$.
Hence, it follows that
\begin{equation*}
	\expct{\ALG(\INS)} \leq \expct{c_1(F_1) + c_2(F_2)} \leq \frac{1+\eps}{1-\eps} \cdot \OPT^{nr}(\INS) = (1 + O(\eps))\cdot \OPT(\INS),
\end{equation*}
achieving the desired ratio.

\noindent\textbf{Analysis of Running Time.}
As mentioned above,
if $\heur^{(i)}_u(\INS)$ is found to be critical,
then in the decomposed sub-instances $\INS_1$ and $\INS_2$,
$\heur^{(i)}_u(\INS_2)$ should be small.
Hence, it follows that there will be at most $|X| \cdot L = \poly(n)$
recursive calls to $\ALG$.
Therefore, as far as obtaining polynomial running times,
it suffices to analyze the running time of the dynamic program $\DP$.
The details are in Section~\ref{section:dp}.

\ignore{

\begin{definition}(Doubling Dimension)
\end{definition}
\begin{definition}(Hierarchical Nets)
\end{definition}
\begin{definition}(Net-Respecting Solution)
\end{definition}
In the following, we will assume that $S$ has bounded doubling dimension and let $k$ be the doubling dimension of the instance $S$. The cost of an edge $e$ will be denoted as $w(e)$ and the penalty of each point $v$ in $S$ will be denoted as $\pi_S(v)$. If $W$ is an instance of \PCX, then $\OPT(W)$ denotes the optimal \PCX solution of $W$, and $\OPTnr(W)$ denotes the optimal net-respecting \PCX solution on $W$. Suppose $T$ is a solution in the original instance $S$, and $B\subset S$. Denote $T\sub{B}$ the subset of edges of $T$ with both endpoints in $B$, and $T\backslash B$ the rest of the edges in $T$. Denote $S\sub{B}$ the sub-instance of \PCX formed by the points of $S$ that in $B$, with penalty exactly the same as in $S$. We will also adopt the notation of $B$ to denote $S\sub{B}$, when its meaning is clear from the context. As for cost, if $T$ is a \PCX solution on an instance $W$, then $\wT_W(T)$ denotes the total edge cost and $\wP_W(T)$ denotes the total penalty $T$ has to pay for the terminal points it fails to visit. Let $w_W(T)=\wT_W(T)+\wP_W(T)$. Throughout the paper, we define the constant $q_0:=O(\frac{k^2s}{\epsilon})^O(k)$.
} 

\section{Sparsity Estimator for \PCX}
\label{section:estimator}

Recall that in the framework outlined in Section~\ref{section:prelim},
given an instance $W$ of $\PCX$,
we wish to estimate the weight of $\OPT^{nr}(W)|_{B(u, 3s^i)}$
with some heuristic $\heur^{(i)}_u(W)$.
We consider a more general sub-instance associated with
the ball $B(u, t s^i)$ for $t \geq 1$.

\noindent\textbf{Auxiliary Sub-Instance.}  Given an instance $W = (T, \pi)$,
$i \in [L]$, $u \in N_i$ and $t \geq 1$,
the sub-instance $W^{(i,t)}_u$ is characterized by terminal set
$W \cap B(u, t s^i)$, equipped with penalties given by the same $\pi$.
Using the classical (deterministic) $2$-approximation algorithms by Goemans and Williamson
for $\PCX$~\cite{DBLP:journals/siamcomp/GoemansW95}, we obtain a
$2$-approximation and then make it net-respecting to produce solution $F^{(i,t)}_u$,
which has cost $c(F^{(i,t)}_u) \leq 2(1 + O(\eps)) \cdot \OPT(W^{(i,t)}_u)$.

\noindent\textbf{Defining the Heuristic.}
The heuristic is defined as $\heur^{(i)}_u(W) := c(F^{(i,4)}_u)$.
\ignore{
\hubert{Explain why the heuristic includes the penalties as well.}
\shaofeng{This is somehow covered in the intro, although what we are saying in the intro is that cost vs cost is not possible, but cost vs tour weight is possible.}
}


In order to show that the heuristic gives a good upper bound on the local sparsity
of an optimal net-respecting solution, we need the following structural result
in Proposition~\ref{prop:long_chain}~\cite[Lemma~3.2]{DBLP:conf/focs/ChanHJ16}
on the existence of long chain in well-separated terminals in a Steiner tree.
As we shall see, the corresponding argument for the case $\PCTSP$ is trivial.

Given an edge set $F$, a \emph{chain} in $F$
is specified by a sequence of points
$(p_1, p_2,\ldots, p_l)$ such that
there is an edge $\{p_i, p_{i+1}\}$ in $F$ between adjacent points,
and the degree of an internal point $p_i$ (where $2 \leq i \leq l-1$)
in $F$ is exactly 2.


\begin{proposition}[Well-Separated Terminals Contains A Long Chain]
	\label{prop:long_chain}
Suppose $S$ and~$T$ are sets in a metric space
with doubling dimension at most $k$ such that $\Diam(S\cup T) \leq D$, and
$d(S, T) \geq \tau D$, where $0<\tau<1$.
Suppose $F$ is an optimal net-respecting Steiner tree
covering the terminals in $S \cup T$.
Then, there is a chain in $F$ with weight at least $\frac{\tau^2}{4096 k^2} \cdot D$ such that
any internal point in the chain is a Steiner point.
\end{proposition}

\begin{lemma}[Local Sparsity Estimator]
\label{lemma:sparsity_estimator}
Let $F$ be an optimal net-respecting solution
for an instance $W$ of $\PCX$.
Then, for any $i \in [L]$, $u \in N_i$ and $t \geq 1$,
we have

$w(F\sub{B(u, t s^i)}) \leq c(F^{(i, t+1)}_u) + O(\frac{skt}{\eps})^{O(k)} \cdot s^i$.
 %
\end{lemma}

\begin{proof}
We follow the proof strategy in~\cite[Lemma~3.3]{DBLP:conf/focs/ChanHJ16},
except that now a feasible solution needs not visit all terminals and
can incur penalties instead.
We denote $B := B(u, t s^i)$ and $\widehat{B} := B(u, (t+1) s^i)$.

Given an optimal net-respecting solution $F$ for instance $W$ of $\PCX$,
we shall construct another net-respecting solution
in the following steps.

\begin{compactitem}
\item[1.] Remove edges in $F|_{B}$.

\item[2.] Add edges $F^{(i,t+1)}_u$
 corresponding to some approximate solution to the instance
$W^{(i,t+1)}_u$ restricted to the ball $\widehat{B}$.

\item[3.]  Let $\eta := \Theta(\frac{\eps}{(t+1)k^2})$,
where the constant in Theta depends on Proposition~\ref{prop:long_chain}.
Let $j$ be the integer such that $s^j \leq \max\{1, \Theta(\frac{\eps}{(t+1)k^2}) \cdot
s^i \} < s^{j+1}$.

Add edges in a minimum spanning tree $H$
of $N_j \cap B(u, (t+2)s^i)$ and edges to connect $H$ to
$F^{(i,t+1)}_u$.

Convert each added
edge into a net-respecting path if necessary.
Observe that the weight of edges
added in this step is $O(\frac{s t k}{\eps})^{O(k)} \cdot s^i$.

\item[4.] So far we have accounted for every terminal inside $\widehat{B}$,
which is either visited or charged with its penalty according
to $c(F^{(i,t+1)}_u)$.  We will give a more detailed
description to ensure that the terminals outside $\widehat{B}$
that are covered by $F$ will still be covered by the new solution.

For $\PCTSP$, we will show that this step can be achieved by increasing
the weight by at most $O(\frac{s t k}{\eps})^{O(k)} \cdot s^i$;
for $\PCSTP$, this can be achieved by
replacing some edges without increasing the weight.
\end{compactitem}

Hence, after the claim in Step 4 is proved,
the optimality of $F$ implies the result.

\noindent\textbf{Ensuring Terminals Outside $\widehat{B}$ are accounted for.}  We achieve
this by considering the following steps.
\begin{enumerate}
\item Consider a connected component $C$ in $F \setminus (F|_B)$. Recall that
the goal is to make sure that all terminals outside $\widehat{B}$ that are
visited by $C$ will also be visited in the new solution.
\item Pick some $x$ in $C \cap B$.  If no such $x$ exists, this implies
that we have the trivial situation $F|_B = \emptyset$.  Let $\widehat{C} \subseteq C$
be the maximal connected component containing $x$ that is contained within $\widehat{B}$.
Define $S := \widehat{C} \cap B$ (which contains $x$)
and $T := \{y \in \widehat{C} \cap \widehat{B}: \exists v \notin \widehat{B}, \{y,v\} \in F\}$,
which corresponds to the points that are connected to the outside $\widehat{B}$.
\ignore{\haotian{@Shaofeng: I feel it is clearer to explain these sets. E.g. $T$ corresponds to the points that are connected to the outside.}},
which is the set of vertices in $\widehat{C}$ that are directed connected by $F$ to some
point outside $\widehat{B}$. Again, the case that $T = \emptyset$ is trivial.
\end{enumerate}

\noindent\textbf{Case (a): There exists $y \in T$, $d(u,y) \leq (t+ \frac{1}{2}) s^i$.}
In this case, this implies there is some $v \notin \widehat{B}$
such that $\{y,v\} \in F$ and $d(y,v) \geq \frac{s^i}{2}$.
Since $F$ is net-respecting, this implies that $y \in N_j$ and hence,
the component $\widehat{C}$ (and also $C$) is already connected to $H$.

\noindent\textbf{Case (b): For all $y \in T$, $d(u,y) > (t+ \frac{1}{2}) s^i$.}
We next show that there is a long chain contained in $\widehat{C}$.
For $\PCTSP$, this is trivial, because we know that $T$ contains only $y$,
and $\widehat{C}$ is a chain from $a=x$ to $b=y$ of length at least $d(x,y) \geq \frac{s^i}{2}$.

For $\PCSTP$, by the optimality of $F$, it follows that
$\widehat{C}$ is an optimal net-respecting Steiner tree covering
vertices in $S \cup T$.  Hence, using Proposition~\ref{prop:long_chain},
$\widehat{C}$ contains some chain from $a$ to $b$ with length
at least $4 \eta s^i$ (where the constant in the Theta in the
definition of $\eta$ is chosen such that this holds).

Once we have found this chain from $a$ to $b$, we remove the edges in this chain.
Hence, we can use this extra weight to connect $a$ and $b$ to their corresponding closest points in $N_j$
via a net-respecting path; observe that for $\PCTSP$, it suffices
to connect only $b=y$ to it closest point in $N_j$.

Finally, observe that for $\PCTSP$,
it is possible to carry out the above procedures such that
all vertices with odd degrees are in the minimum spanning tree $H$
\ignore{\haotian{@Shaofeng: Probably we should explain this in more details? For example, to say that we need to apply Christofide's approach and we only want to use edges in $H$ more than once?}}.
Therefore, extra edges are added to ensure that the degree of every vertex is even
to ensure the existence of an Euler circuit. This has extra cost at most
$w(H) \leq O(\frac{s t k}{\eps})^{O(k)} \cdot s^i$. This completes the proof.

\end{proof}

\vspace{-10pt}
\begin{corollary}[Threshold for Critical Instance]
\label{cor:threshold}
	Suppose $F$ is an optimal net-respecting solution for an instance $\INS$ of \PCX,
	and $q \geq \Theta(\frac{s k}{\epsilon})^{\Theta(k)}$.
	If for all $i\in [L]$ and $u\in N_i$, $\heur^{(i)}_u(\INS) \leq q s^i$, then $F$ is $2q$-sparse.
\end{corollary}

\section{Decomposition into Sparse Instances}
\label{section:decomposition}

In Section~\ref{section:estimator},
we define a heuristic $\heur^{(i)}_u(\INS)$
to detect a critical instance around some point $u \in N_i$
at distance scale $s^i$.  We next describe
how the instance $\INS$ of $\PCX$ can be decomposed into
$\INS_1$ and $\INS_2$ such that equations~(\ref{eq:directed_union})
and~(\ref{eq:decompose})
in Section~\ref{sec:sparse_overview}
are satisfied.


\noindent \textbf{Decomposing a Critical Instance.}
We define a threshold $q_0 := \Theta(\frac{s k}{\epsilon})^{\Theta(k)}$ according to 
Corollary~\ref{cor:threshold}.  As stated
in Section~\ref{sec:sparse_overview},
a critical instance is detected by the heuristic
when a smallest $i \in [L]$ is found
for which there exists some $u \in N_i$
such that $\heur^{(i)}_u(\INS) = c(F^{(i,4)}_u) > q_0 s^i$
\ignore{\haotian{@Shaofeng: Have we defined $q_0$ before?}}.
Moreover, in this case, $u \in N_i$ is chosen
to maximize $\heur^{(i)}_u(\INS)$.
To achieve a running time with an
$\exp(O(1)^{k \log(k)})$ dependence
on the doubling dimension $k$,
we also apply the technique in~\cite{DBLP:conf/soda/ChanJ16} to
choose the cutting radius carefully.

\begin{claim}[Choosing Radius of Cutting Ball]
\label{claim:cut_ball}
Denote $\T(\lambda) := c(F^{(i, 4 + 2\lambda)}_u)$.
Then, there exists $0 \leq \lambda < k$ such that
$\T(\lambda+1) \leq 30 k \cdot \T(\lambda)$.
\end{claim}

\begin{proof}
A similar proof is found in~\cite{DBLP:conf/focs/ChanHJ16},
and we adapt the proof to include penalties of unvisited terminals.
Suppose the contrary is true.
Then, it follows that $\T(k) > (30k)^k \cdot \T(0)$.
We shall obtain a contradiction by showing that 
there is a solution for the instance $\INS^{(i,4+2k)}_u$ corresponding to $\T(k) = c(F^{(i,4+2k)}_u)$ with small weight.
Define $N_{i}'$ to be the set of points in $N_i$ that cover $B(u, (2k+5)s^i)$.


We construct an edge set $F$ that is a solution to
the instance $\INS^{(i,4+2k)}_u$.
For each $v \in N_i'$, we include
the edges in the solution $F^{(i,4)}_v$,
whose cost includes the edge weights and the penalties of unvisited terminals.
By the choice of $u$, the sum of the costs of these partial solutions
is at most $|N_i'| \cdot \T(0)$.

We next stitch these solutions together by adding
extra edges of total weight at most $2 \cdot 2(2k+5) \cdot |N_i'| \cdot s^i$;
for \PCTSP, we make sure that the degree of every vertex is even
to form an Euler tour.

Hence, the resulting solution $F$ has cost
$c(F) \leq 4(2k+5) |N_i'| \cdot s^i + |N_i'| \cdot \T(0) \leq  (15 k)^k \cdot \T(0)$.
Therefore, we have an upper bound
for the heuristic $\T(k) \leq 2(1 + \Theta(\eps)) \cdot w(F) \leq (30)^k \cdot \T(0)$,
which gives us the desired contradiction.
\end{proof}

\noindent \textbf{Cutting Ball and Sub-Instances.}
Suppose $\lambda \geq 0$ is picked as in Claim~\ref{claim:cut_ball},
and sample $h \in [0, \frac{1}{2}]$ uniformly at random.
Define $B := B(u, (4 + 2 \lambda + h)s^i)$.
The original instance $\INS = (T, \pi)$ is decomposed into instances $\INS_1$ and $\INS_2$ as follows:

\begin{compactitem}
\item For $\INS_1 = (T_1, \pi_1)$, the terminal set is $T_1 := (B \cap T) \cup \{u\}$, 
where for $v \neq u$ $\pi_1(v) := \pi(v)$ and $\pi_1(u) := + \infty$.
We denote the cost function associated with $\INS_1$ by $c_1$.
  
\item Suppose ${F}_1$ is the (random) solution for instance $\INS_1$ (that covers $u$)
returned by the dynamic program for sparse instances in Section~\ref{section:ptas_sparse}.
Then, instance $\INS_2 = (T_2, \pi_2)$ is defined with respect to ${F}_1$.  The terminal set
is $T_2 := (T \setminus B) \cup \{u\}$.  For $v \in T_2 \setminus \{u\}$, $\pi_2(v) := \pi(v)$
is the same; however, $\pi_2(u) := \pi(T \cap B) - c_1({F}_1) = \pi(T \cap B \cap {F}_1) - w({F}_1)$.
\end{compactitem}

Observe that the instance $W_2$ depends on ${F}_1$ through the choice of the penalty for $u$.

\begin{lemma}[Combining Solutions of Sub-Instances]
\label{lemma:subinstance}
Suppose instance $\INS_1$ is defined with cost function $c_1$
and instance $\INS_2$ is defined with respect to ${F}_1$
of $\INS_1$.  Furthermore, suppose $\widehat{F}_2$ is a solution to instance $\INS_2$, whose
cost function is denoted as $c_2$.
Then, we have the following.

\begin{compactitem}
\item[(i)] Suppose $\widehat{F}_1$ is any solution to $\INS_1$ that contains $u$,
and let $F := \widehat{F}_2 \looparrowleft_u \widehat{F}_1$.
If $\widehat{F}_2$ covers $u$, then $F = \widehat{F}_2 \cup \widehat{F}_1$ is a solution
to $\INS$ with cost $c(F) \leq c_1(\widehat{F}_1) + c_2(\widehat{F}_2)$;
if $F_2$ does not cover $u$, then $F = \widehat{F}_2$ is a solution
to $\INS$ with cost $c(F) \leq c_1({F}_1) + c_2(\widehat{F}_2)$.
This implies~(\ref{eq:directed_union}) in Section~\ref{sec:sparse_overview}.
\ignore{
\haotian{Can we replace the inequality sign with equality?}}

\item[(ii)] The sub-instance $\INS_2$ does not have a critical instance
with height less than $i$, and $\heur^{(i)}_u(\INS_2) = 0$.

\item[(iii)] $\heur^{(i)}_u(\INS_1) \leq O(s)^{O(k)} \cdot q_0 \cdot s^i$.
\end{compactitem}
\end{lemma}

\begin{proof}
For the first statement,
Definition~\ref{defn:extension}
ensures that $F$ is connected;
for $\PCTSP$, it suffices to observe that the union
of two intersecting tours is also a tour.  Hence, $F$ is a feasible solution
for the instance $\INS$ of $\PCX$.

We next give an upper bound on $c(F)$, by pessimistically considering the case that $\widehat{F}_2$
does not cover any terminal in $B \cap T$.

For the case that $\widehat{F}_2$ covers $u$, we have $F= \widehat{F}_2 \cup \widehat{F}_1$ and
we have $c(F) = w(\widehat{F}_1) + w(\widehat{F}_2) + \pi(T \setminus F)
\leq w(\widehat{F}_1) + \pi(T_1 \setminus \widehat{F}_1) + w(\widehat{F}_2) + \pi((T \setminus B) \setminus \widehat{F}_2)
\leq c_1(\widehat{F}_1) + c_2(\widehat{F}_2)$.

For the case that $\widehat{F}_2$ does not cover $u$, we have $F = \widehat{F}_2$, 
and
$c(F) = w(\widehat{F}_2) + \pi(T \setminus \widehat{F}_2) \leq w(\widehat{F}_2) +  \pi((T \setminus B) \setminus \widehat{F}_2)
+ \pi(T \cap B) \leq w(\widehat{F}_2) +  \pi((T \setminus B) \setminus \widehat{F}_2)
+ \pi_2(u) + c_1({F}_1) = c_2(F_2) + c_1({F}_1)$.

The second statement follows from the choice of~$i$.  Moreover,
$\heur^{(i)}_u(\INS_2) = 0$ because in instance $W_2$ the only terminal in $B$ is
$u_1$, which can be covered by a self-loop of weight 0.

For the third statement,
we use the fact that there is no critical instance at height $i-1$
to show that there is a solution to $\INS_1$ with small cost.

Moreover, we consider the solutions corresponding
to $\heur^{(i-1)}_v(\INS)$, over $v \in N_{i-1} \cap B(u, 5s^i)$.
The cost is $O(s)^{O(k)} \cdot q_0 \cdot s^i$.

In order to stitch these partial solutions together,
we add extra edges with weights at most $|N_{i-1}| \cdot O(s^i)$.  Hence,
the total cost for the solution to (any sub-instance of) $\INS_1$ is at most
$O(s)^{O(k)} \cdot q_0 \cdot s^i$.
%
%
%
%
\end{proof}

\begin{lemma}[Combining Costs of Sub-Instances]
\label{lemma:combine_cost}
	Suppose $F$ is an optimal net-respecting solution for instance $\INS$ of $\PCX$.
	Then, for any realization of the decomposed sub-instances $\INS_1$ and $\INS_2$
	as described above, there exist (not necessarily net-respecting) solution $\widehat{F}_1$ for $\INS_1$ and
	net-respecting solution $\widehat{F}_2$ for $\INS_2$
	such that
	$(1-\epsilon) \cdot \expct{c_1(\widehat{F}_1)} + \expct{c_2(\widehat{F}_2)} \leq c_W(F)$, where the expectation
	is over the randomness to generate $\INS_1$ and $\INS_2$.
	Recall that the randomness to generate $\INS_1$ and $\INS_2$ involves the random ball $B$
	and the randomness used in the dynamic program to generate $F_1$ to produce
	instance $\INS_2$ and its cost function $c_2$.
\end{lemma}

\begin{proof}
Recall that the random ball $B = B(u, (4 + 2 \lambda + h) \cdot s^i)$
for random $h \in [0, \frac{1}{2}]$,
and denote $\overline{B} := B(u, (4 + 2 \lambda + 1) \cdot s^i)$, which
is deterministic.
For the trivial case $V(F) \cap \overline{B} = \emptyset$,
we choose $\widehat{F_1} := F_1$ (which is the solution used
to define $W_2$ and $c_2$) and $\widehat{F_2} := F$.
In this case, we have $c_1(\widehat{F}_1) + c_2(\widehat{F}_2) =
c_1(F_1) + (\pi(T \cap B) - c_1(F_1)) + \pi((T \setminus B) \setminus F) + w(F) =  c(F)$.

For the rest of the proof, we can assume
that $V(F) \cap \overline{B}$ is non-empty.  Moreover,
the solution $\widehat{F}_2$ we are going to construct
will always include $u$.


Denote $\widehat{V_1} := \{x \in {B} \mid 
\exists y \notin {B}, \{x,y\} \in F\}$.

We start by including $F|_{B}$ in $\widehat{F}_1$,
and including the remaining edges of $F$ in $\widehat{F}_2$.
Then, we will add extra edges such that (i) in each of $\widehat{F}_1$ and
$\widehat{F}_2$, the vertices covered form a connected component
and include $\widehat{V}_1$
\ignore{\haotian{Have we defined $\widehat{V}_1$ before?}}, (ii) $\widehat{F}_2$ visits $u$,
(iii) for $\PCTSP$, every vertex has even degree.

Hence, all the terminals in $V(F) \cap B$ are visited by $\widehat{F}_1$
and all terminals in $V(F) \setminus B$ are visited by $\widehat{F}_2$.
If we can show that these extra edges have expected total weight
at most $\epsilon \cdot \expct{c_1(\widehat{F}_1)}$,
then the lemma follows.

Define $N$ to be the subset of $N_j$ that cover the points in $\overline{B}$,
where $s^j <  \delta s^i \leq s^{j+1}$ and
$\delta = \Theta(\frac{\epsilon}{k})$.
We include edges of a minimum spanning tree $H$ of $N$ in each of $\widehat{F}_1$ and $\widehat{F}_2$, and make it net-respecting; for $\PCTSP$, each edge in $H$ can be included a constant number of times
to ensure that the degree of every vertex is even.
Furthermore, since $V(F) \cap \overline{B}$ is non-empty,
even when $V(F) \cap B$ is empty, it just takes one net-respecting path of length at most $2 \delta s^i$
to connect $\widehat{F}_2$ to $H$.
The sum of the weights of edges added from $H$ is at most $|N|\cdot O(k) \cdot s^i \leq O(\frac{k s}{\epsilon})^{O(k)} \cdot s^i$. 


\noindent\textbf{Connecting Crossing Points.}
Recall that $\widehat{V_1} := \{x \in {B} \mid 
\exists y \notin {B}, \{x,y\} \in F\}$.
To ensure the connectivity of both $\widehat{F}_1$ and
$\widehat{F}_2$,
we add extra edges to ensure that in each of
$\widehat{F}_1$ and $\widehat{F}_2$,
each point in $\widehat{V_1}$ is connected to some point in $N$,
which is connected by edges in~$H$.

Note that if
such a point $x \in \widehat{V_1}$ is incident to some edge in $F$
with weight at least $\frac{s^i}{4}$,
then the net-respecting property of $F$ implies
that $x$ is already in $N$.  Otherwise, we need to connect
$x$ to some point in $N$ with a net-respecting path of length at most $2 \delta s^i$;
observe that this happens because some edge $\{x,y\}$ in $F$
is cut by $B$, which
happens with probability at most $O(\frac{d(x,y)}{s^i})$.
Hence, each edge $\{x,y\} \in F|_{\overline{B}}$
has an expected contribution of $\delta s^i \cdot O(\frac{d(x,y)}{s^i}) = O(\delta) \cdot d(x,y)$.

\noindent\textbf{Charging the Extra Costs to $\widehat{F}_1$.}
Apart from using edges in $F$,
the extra edges come from a constant number of copies of the
minimum spanning tree $H$,
and other edges with cost $O(\delta) \cdot w(F|_{\overline{B}})$.  We charge these extra costs
to $c_1(\widehat{F}_1)$.

Observe that the heuristic $c(F^{(i,4)}_u) > q_0 \cdot s^i$ and $\widehat{F}_1$ is a solution for $\INS^{(i,4+2\lambda+h)}_u$.
Therefore, by the definition of $\heur^{(i)}_u(W)$,
we have
$c_1(F_1) \geq \frac{1}{2(1+\Theta(\eps))} \cdot c(F^{(i,4)}_u) \geq
\frac{q_0}{8} \cdot s^i$,
by choosing large enough $q_0$.
Therefore, the sum of weights of edges from $H$ is at most $O(\frac{k s}{\epsilon})^{O(k)} \cdot s^i
\leq \frac{\epsilon}{2} \cdot c_1(F_1)$.

We next give an upper bound on $w(F|_{\overline{B}})$,
which is at most 
$c(F^{(i,4+2(\lambda+1))}_u) + O(\frac{s k}{\epsilon})^{O(k)}
\cdot s^i$, by Lemma~\ref{lemma:sparsity_estimator}.
By the choice of $\lambda$,
we have $c(F^{(i,4+2(\lambda+1))}_u) \leq 30k \cdot 
c(F^{(i,4+2\lambda)}_u)$.
\ignore{\hubert{The corresponding inequality
in the FOCS paper has a typo.}}

Moreover, since $\widehat{F}_1$ is also a solution
for $W^{(i,4+2\lambda)}_u$,
$c(F^{(i,4+2\lambda))}_u) \leq 2(1+\Theta(\eps)) \cdot c_1(\widehat{F}_1)$.
Hence, we can conclude that
$w(F|_{\overline{B}}) \leq O(k) \cdot c_1(\widehat{F}_1)$.

Hence, by choosing small enough $\delta = \Theta(\frac{\eps}{k})$,
we can conclude that the extra edges has expected weight at most
$O(\delta) \cdot w(F|_{\overline{B}})
\leq \frac{\eps}{2} \cdot c_1(\widehat{F}_1)$.

Therefore, we have shown that
$\expct{c_1(\widehat{F}_1)} + \expct{c_2(\widehat{F}_2)}
\leq c(F) + \eps \cdot c_1(\widehat{F}_1)$, where the right hand side is a random variable.
Taking expectation on both sides and rearranging gives
the required result.
\end{proof}

\section{Revisiting Hierarchical Decomposition and Sparse Instance Frameworks
for TSP-like Problems}
\label{section:ptas_sparse}

In this section, we revisit the randomized hierarchical framework that is used
in all known PTAS's (and QPTAS's) for TSP-like problems in doubling metrics~\cite{DBLP:conf/stoc/Talwar04,DBLP:journals/siamcomp/BartalGK16,DBLP:conf/soda/ChanJ16,DBLP:conf/focs/ChanHJ16}.
As mentioned in Section~\ref{section:intro} and Remark~\ref{remark:why_fix}, in the original paper~\cite{DBLP:conf/stoc/Talwar04},
the randomness in the underlying hierarchical decomposition is first used
to bound the increase in the cost of a solution to achieve some \emph{portal-respecting property}.
However, conditioned on the portal-respecting property, some more careful arguments should be required for the conditional randomness of the hierarchical decomposition.

Since this random hierarchical framework is widely used in subsequent works, we think it is worthwhile to revisit the framework and resolve any previous technical issues.  In particular,
in Section~\ref{sec:hier_decomp} we give a more precise definition and notation for cluster portals in a hierarchical decomposition.
As a result of the modified definition of portals, in Section~\ref{section:struct_proof}, we also revisit the
analysis of the sparsity framework~\cite{DBLP:journals/siamcomp/BartalGK16}
that was used to achieve the first PTAS for TSP on doubling metrics.
Even though we use similar concepts in the modified framework,
some arguments are quite different from previous proofs.
In particular, below are highlights of the changes made in the modified framework:

\begin{compactitem}
	\item We make use of a net tree to define portals with respect to a hierarchical decomposition.  As a result, we also need to modify the notion of $(m, r)$-lightness for a solution.
	
	\item As opposed to previous approaches~\cite{DBLP:conf/stoc/Talwar04,DBLP:journals/siamcomp/BartalGK16}, when a solution uses too many active portals for a cluster, our patching argument does not rely on the small MST lemma \cite[Lemma 6]{DBLP:conf/stoc/Talwar04}.
	
	\item After a given solution is modified to observe the portal-respecting property,
	any newly edges actually depend on the randomness of the hierarchical decomposition.  Hence, in order to use the randomness of the decomposition again, we give a new charging argument that, loosely speaking, maps a newly added edge back to an original edge that created it.
\end{compactitem}

\ignore{

Suppose the instance $\INS$ has a $q_0$-sparse optimal net-respecting solution.
In this section, we shall prove a structural theorem for such instances.

We follow the hierarchical decomposition framework introduced and used in~\cite{DBLP:conf/stoc/Talwar04,DBLP:journals/siamcomp/BartalGK16,DBLP:conf/soda/ChanJ16,DBLP:conf/focs/ChanHJ16}.
We start with a review of the framework in Section~\ref{sec:hier_decomp}, and in particular, we consider $(m, r)$-light solutions.

Then in Section~\ref{section:struct_proof}, we present the structural theorem, whose analogue appeared in several previous works(e.g. ~\cite{DBLP:conf/stoc/BartalGK12}).
In fact, this theorem is very similar to some previous ones, and at a first glance, there is no need to supply a proof here.
However, as mentioned in Remark~\ref{remark:why_fix}, several issues remain in the previous proofs. Our proof gives a formal treatment of all those issues. To name a few notable changes:
}

\subsection{Randomized Hierarchical Decomposition Framework}
\label{sec:hier_decomp}

\noindent\textbf{Net Tree.}
Recall that given a metric space, we consider a sequence
of hierarchical nets $\{N_i\}_i$ as defined in Section~\ref{section:prelim}.
We define a \emph{net tree} with respect to the hierarchical nets $\{N_i\}_i$
in a way similar to~\cite{ChanLNS15}; for notational convenience,
we assume that for all $i \leq 0$, $N_i = X$.
For each height~$i$ and each $u \in N_i$, there is some node~$(u,i)$ in the net tree;
for notational convenience, for $i < 0$, $(u,i) := (u,0)$.  The metric $d$
can naturally be extended to nodes.
Denote $\widehat{N}_i := \{(u, i) \mid u \in N_i\}$, and the
tree has node set $\widehat{X} := \bigcup_{i}{ \widehat{N}_i}$.
Notice we use point to refer to an element in~$X$ and a node to refer to an element in~$\widehat{X}$.
Observe that $N_L$ contains only one point $r \in X$, and the corresponding node $(r, L)$
is the root of the net tree.
The edges of the tree is defined by a parent function $\Par$, mapping a non-root node to its parent. For $i < L$ and $u \in N_i$, define $\Par(u, i) := (v, i+1)$,
where $v\in N_{i+1}$ is the closest point in $N_{i+1}$ to $u$ (breaking ties arbitrarily).
For a point $u \in X$, define $\Anc_j(u) \in \widehat{N}_j$ be the height-$j$ ancestor of $(u,0)$.
In this section, we assume an underlying net tree is constructed.

\noindent \textbf{Subgraph on Nodes.}  Observe that a multi-graph $\widehat{G}$ with vertex set
in~$\widehat{X}$ naturally induces a multi-graph~$G$ with vertex set in~$X$. Every edge $\{(u,i), (v,j)\}$ in $\widehat{G}$ induces an edge $\{u,v\}$ in~$G$ if $u \neq v$.  Recall that
we consider multi-graphs because the solution for \PCTSP needs to be Eulerian.

We use the following decompositions as mentioned in~\cite{DBLP:journals/siamcomp/BartalGK16,DBLP:conf/soda/ChanJ16,DBLP:conf/focs/ChanHJ16}.

\begin{definition}[Single-Scale Decomposition~\cite{DBLP:conf/stoc/AbrahamBN06}]
	\label{defn:single_decomp}
	At height $i$, an arbitrary ordering $\pi_i$ is
	imposed on the net $N_i$.  Each net-point $u \in N_i$
	corresponds to a \emph{cluster center} and
	samples random $h_u$ from a truncated exponential distribution
	$\Exp_i$ having density function $t \mapsto
	\frac{{\chi}}{{\chi}-1} \cdot \frac{\ln \chi}{s^i} \cdot e^{-\frac{t \ln \chi }{s^i}}$ for $t \in [0, s^i]$, where $\chi = O(1)^k$.
	Then, the cluster at $u$ has random radius $r_u := s^i + h_u$.

	The clusters
	induced by $N_i$ and the random radii form a
	decomposition $\Pi_i$,
	where a point $p \in V$ belongs to the cluster
	with center $u \in N_i$ such that $u$ is the first
	point in $\pi_i$ to satisfy $p \in B(u, r_u)$.
	We say that the partition $\Pi_i$ cuts a set $P$
	if $P$ is not totally contained within a single cluster.
	
	The results in~\cite{DBLP:conf/stoc/AbrahamBN06} imply that
	the probability that a set $P$ is cut by $\Pi_i$
	is at most $\frac{\beta \cdot \Diam(P)}{s^i}$,
	where $\beta = O(k)$.
\end{definition}

\begin{definition}[Hierarchical Decomposition] \label{defn:phd}
	Given a configuration of random radii
	for $\{N_i\}_{i \in [L]}$, decompositions $\{\Pi_i\}_{i \in [L]}$
	are induced as in Definition~\ref{defn:single_decomp}.
	At the top height $L-1$, the whole space
	is partitioned by $\Pi_{L-1}$ to form height-$(L-1)$ clusters.  Inductively,
	each cluster at height $i+1$ is partitioned
	by $\Pi_i$ to form height-$i$ clusters, until height $0$ is reached.  Observe that a cluster
	has $K := O(s)^k$ child clusters.
	Hence, a set $P$ is cut at height $i$ \emph{iff}
	the set $P$ is cut by some partition $\Pi_j$ such that
	$j \geq i$; this happens with probability
	at most $\sum_{j \geq i} \frac{\beta \cdot \Diam(P)}{s^i} = \frac{O(k) \cdot \Diam(P)}{s^i}$.
\end{definition}

\noindent\textbf{Portals.} We define portals with respect to some hierarchical decomposition.
For a height-$i$ cluster $C$, define its portals as $\{ \Anc_j(u) \mid u \in C \}$, where $j$ satisfies $s^j \leq \Theta(\frac{\eps}{kL}) \cdot s^i < s^{j+1}$.
Observe that the same node in $\widehat{X}$ could be a portal for several clusters
of the same height.  However, we emphasize that
 a portal~$p$ is naturally associated with some cluster that it is assigned.
Hence, whenever we talk about a portal $p$, we implicitly mean that ``$p$ is a portal of some cluster $C$ of height $i$'',
and say that $p$ is a height-$i$ portal for short.
We use $\widehat{P}_i$ to denote the set of height-$i$ portals,
and denote $\widehat{P} := \cup_i \widehat{P}_i$.

Since a height-$i$ cluster
has diameter $O(s^i)$, by packing property,
each cluster has at most $m := O(\frac{k L s}{\epsilon})^k$
portals.

\noindent\emph{Solutions on Portals.} Observe that a multi-graph $G$ with vertex~$\widehat{P}$
naturally induces a multi-graph with vertex set~$\widehat{X}$ (and a multi-graph with vertex set~$X$) in the natural way.
For a terminal $t \in X$, a multi-graph $G$ solution
on $\widehat{P}$ visits $t$ only \emph{iff} $G$ covers $(t,i)$ for some $i$.

\noindent\textbf{Portal-Respecting Solution.}
Our algorithm works with solutions with vertex set~$\widehat{P}$.
A multi-graph $F$ with vertex set in~$\widehat{P}$, is called portal-respecting with respect to some hierarchical decomposition,
if for any edge $e = \{u, v\}$ in $F$, where $u$ is a portal of height-$j$ cluster $C$ and $v$ is a portal of height-$j'$ cluster $C'$ with $j\geq j'$,
it holds that
\begin{compactitem}
	\item If $j = j'$, then $C$ and $C'$ have the same parent cluster;
	\item If $j > j'$, then $j = j' + 1$ and $C'$ is a child cluster of $C$.
\end{compactitem}

\noindent\textbf{Active Portals.}
Suppose $F$ is portal-respecting (with respect to some hierarchical decomposition).
Consider a portal $p$ of a height-$i$ cluster $C$ that is visited by $F$.
We say that $p$ is an \emph{active} portal if $p$ is connected (in $F$) to a height-$i$ portal of a sibling of $C$, or a height-$(i+1)$ portal of a parent cluster of $C$.


\noindent\textbf{$(m,r)$-Light Solution.}
A (multi-)graph $F$ is called $(m,r)$-\emph{light},
if it is portal-respecting for a hierarchical
decomposition in which each cluster has at most $m$ portals,
and each cluster has at most $r$ active portals.

\begin{remark}
Almost all previous works consider a solution as a subgraph with vertex
set in the original metric space~$X$.  However, such solutions in the
previous frameworks are implicitly induced by ones with portals $\widehat{P}$ as the vertex set.

We have a unified notion of $(m,r)$-lightness that is the same for both \PCTSP and
\PCSTP.  We next describe additional properties for a \PCTSP solution that justify
why our lightness notion does not count the number of times a tour visits a cluster.
\end{remark}

\noindent\textbf{Additional Structure for \PCTSP.} We consider additional
properties of a portal-respecting solution, which is Eulerian.

\begin{definition}[Crossing Portal Pair]
\label{defn:crossing_pair}
A portal-respecting Eulerian tour crosses a cluster $C$ through
the (ordered) portal pair $(p,q)$ (where~$p$ and $q$ can be the same)
if the node immediately preceding $p$ and the node immediately succeeding
$q$ are portals of the parent or a sibling of $C$, and all the nodes (if any) visited
from $p$ to $q$ are portals of (not necessarily proper) descendant clusters of $C$.

A portal-respecting Eulerian tour is economical
if for every cluster $C$ and every ordered portal pair $(p,q)$,
the tour crosses $C$ through $(p,q)$ at most once.
\end{definition}

\begin{definition}[Scratch and Removal]
\label{defn:scratch}
A portal-respecting Eulerian tour scratches a cluster $C$ at portal $p$ if
the two nodes~$x$ and~$y$ that immediately go before and after~$p$ are both portals
of the parent or a sibling of $C$.  Hence, a scratch is a special case of crossing cluster $C$ through $(p,p)$.

Observe that the edge $\{x,y\}$ is portal-respecting.  Hence, if the portal~$p$ is visited
in another part of the tour, the scratch at portal~$p$ of cluster~$C$ can be removed
by using the shortcut $\{x,y\}$, without increasing the length of the tour.
\end{definition}

\begin{lemma}[Economical Tour]
	\label{lemma:unique_pair}
A portal-respecting Eulerian tour can be modified to be economical without increasing its length and still visit the same set of terminals.
\end{lemma}
\begin{proof}
Suppose the tour crosses some cluster $C$ through the ordered pair $(p,q)$ at least twice:
$E_1, p, P_1, q, E_2, p, P_2, q, E_3$, where the $E$'s and $P$'s represent sequences of
visited edges.  Moreover, the nodes visited by the $P$'s are all portals
of the descendant clusters of $C$.
	
Consider an alternative tour:
$E_1, p, P_1, q, \widehat{P_2}, p, \widehat{E_2}, q, E_3$, where $\widehat{S}$ for an edge sequence $S$ denotes the reverse of $S$.
Then, observe that $E_1, p, P_1, q, \widehat{P_2}, p, \widehat{E_2}$ induces only one crossing of $C$ through $(p,p)$. Moreover, the scratch of $C$ at $q$ induced by $\widehat{E_2}, q, E_3$ can be removed as in Definition~\ref{defn:scratch} since $q$ is visited in $E_1, p, P_1, q, \widehat{P_2}, p, \widehat{E_2}$.

Hence, we have replaced two crossings of $C$ at $(p,q)$ by one crossing of $C$ at $(p,p)$. Notice that the above argument holds even in the case where $p=q$. Moreover, the number of edges in the tour is reduced by one due to the removal of the scratch of $C$ at $q$ so the procedure can only be carried out a finite number of times. Using this argument repeatedly gives the result of the lemma.
\end{proof}

\subsection{Sparsity Structural Lemma}
\label{section:struct_proof}

We revisit the sparsity structural lemma
in~\cite[Lemma~3.1]{DBLP:journals/siamcomp/BartalGK16}.
On a high level, the lemma says that given a net-respecting $q$-sparse solution
$T$ and an appropriate hierarchical decomposition,
there exists an $(m,r)$-light solution with appropriate parameters $m$ and $r$,
whose length does not increase too much.

\noindent\textbf{Property of Hierarchical Decomposition.}
Recall that in the hierarchical decomposition, for each $i$ and $u\in N_i$,
a random radius $h^{(i)}_u$ is sampled from a truncated exponential distribution,
and we define a random ball $B(u, r^{(i)}_u)$, where $r^{(i)}_u := s^i + h^{}$.
Let $A^{(i)}_u$ be the event that $B(u, r^{(i)}_u)$ cuts at most $O(q\cdot k)$ edges in $T$ with length at most $s^i$.
The desired property of the hierarchical decomposition is the event $\mathcal{A}$ that all $A^{(i)}_u$ happen simultaneously, for all $i$ and $u\in N_i$.

\begin{proposition}[\cite{DBLP:journals/siamcomp/BartalGK16}]
	\label{prop:cond_prop}
	For any $S\subset X$,
	\begin{equation*}
		\Pr[S \text{ is cut by a height-$j$ cluster} \mid \mathcal{A}] \leq O(k) \cdot \frac{\Diam(S)}{s^j}.
	\end{equation*}
\end{proposition}

\begin{proposition}[\cite{DBLP:journals/siamcomp/BartalGK16}]
	\label{prop:good_radius_constant_prop}
	For any $i$ and $u\in N_i$, $\Pr[A^{(i)}_u] \geq \frac{1}{2}$.
\end{proposition}

\begin{theorem}[Sparse Structural Property]
	\label{theorem:struct}
	Suppose $T$ is an optimal net-respecting solution with points in~$X$ is $q$-sparse .
	Given any hierarchical decomposition,
	there is a way to transform $T$ into an an $(m, r)$-light solution $T'$ on portals $\widehat{P}$ that visits the same terminals as $T$,
	with
	$m := O(\frac{kL s}{\epsilon})^k$,
	$r := q\cdot \Theta(1)^k + \Theta(\frac{s}{\epsilon})^k$,
	such that
	\begin{equation*}
		\Pr[w(T') \leq (1+O(\epsilon))\cdot w(T) \mid \mathcal{A}] \geq \frac{1}{4},
	\end{equation*}
	where the randomness comes from the hierarchical decomposition.
	
	\ignore{
	\hubert{I changed the probability.}
	}
	
	Furthermore, if $T$ is Eulerian, then so is $T'$.
\end{theorem}
\begin{proof}
	Suppose some hierarchical decomposition is fixed.
	
	\noindent\textbf{Part I: Defining Portal-Respecting $T''$.}
	Set $T'' = \emptyset$ initially.
	
	Examine each edge $e:\{u, v\}$ in $T$.
	Let $i$ be the largest height that $e$ is cut,
	and let $C_u$ and $C_v$ be the (unique) height-$i$ clusters that $u$ and $v$ lies in.
	Define $h(e) := i$.
	Include the path
	\begin{align*}
		(  (u,0) = \Anc_0(u), \Anc_1(u), \ldots, \Anc_j(u), \Anc_j(v), \ldots, \Anc_1(v), \Anc_0(v) = (v,0) )
	\end{align*}
	in $T''$, where $j$ satisfies $s^j \leq \Theta(\frac{\epsilon}{kL})\cdot s^i < s^{j+1}$.
	Observe that every node in this path is an active portal, and we say these portals are activated by $e$.
	It is immediate that $T''$ is portal-respecting. Moreover, if $T$ is Eulerian, then $T''$ is Eulerian as well.
	
	For an active portal $p$, let $f(p)$ be any edge in $T$ that activates $p$.
	
	\begin{lemma}
		\label{lemma:portal_resp_cost}
		$E[w(T'') \mid \mathcal{A}] \leq (1+\epsilon) \cdot w(T)$.
	\end{lemma}
	\begin{proof}
		It is sufficient to show that for each edge $e := \{u, v\}$ in $T$,
		the weight of the path
		\begin{equation*}
			P_e := (  (u,0) = \Anc_0(u), \Anc_1(u), \ldots, \Anc_j(u), \Anc_j(v), \ldots, \Anc_1(v), \Anc_0(v) = (v,0) )
		\end{equation*}
		is at most $(1+\epsilon)\cdot w(e)$ in expectation, given $\mathcal{A}$, where $s^j \leq \Theta(\frac{\epsilon}{kL}) \cdot s^i < s^{j+1}$ and $i := h(e)$.
		
		Fix $e := \{u, v\}$ in $T$.
		Let $i := h(e)$ and $j$ defined as in the construction.
		Observe that for any $l$,
		$d(u, \Anc_l(u)) \leq O(s^l)$.
		So, the weight of the path from $u$ to
		$\Anc_j(u)$ is at most
		$O(\frac{\epsilon}{kL}) \cdot s^i$, and so is that from $v$ to $\Anc_j(v)$.
		By triangle inequality, $d(\Anc_j(u), \Anc_j(v)) \leq d(u, v) + O(\frac{\epsilon}{kL}) \cdot s^i$.
		
		Therefore, the additional cost
		\begin{align*}
			\left(\sum_{e' \mathrm{ in } P_e}{w(e')}\right) - d(u, v) \leq O(\frac{\epsilon}{kL})\cdot s^i.
		\end{align*}
		This cost occurs only if height $i$ is the largest height at which $e$ is cut, and it is of probability at most $O(k)\cdot \frac{w(e)}{s^i}$, by Proposition~\ref{prop:cond_prop}.
		Summing this over all $i$, this cost is at most
		\begin{align*}
			\sum_{i\in[L]}{O(\frac{\epsilon}{kL}) \cdot s^i \cdot O(k)\cdot \frac{w(e)}{s^i}} \leq \epsilon \cdot w(e),
		\end{align*}
		in expectation, conditioning on $\mathcal{A}$.  This completes the proof of Lemma~\ref{lemma:portal_resp_cost}.
	\end{proof}
	
	\noindent\textbf{Part II: $(m, r)$-light $T'$.}
	We shall define $T'$ from $T''$, so that $T'$ is $(m, r)$-light.
	Examine each cluster $C$ from higher height to lower height. Let $r'$ be the number of active portals of $C$.
	\begin{compactitem}
		\item If $r' \leq r$, then we do nothing and proceed to the next cluster.
		\item Otherwise $r' > r$.
		Apply the following patching procedure in Lemma~\ref{lemma:patching} to $C$.
	\end{compactitem}
	
	\begin{lemma}[Patching Lemma]
		\label{lemma:patching}
		As defined above, $T''$ is a portal respecting solution.
		Suppose $C$ is a height-$i$ cluster with active portal set $R$. Recall that by definition, an active portal is connected by an edge to a portal of the parent or a sibling cluster of~$C$;
		let $\widehat{R}$ be the set of such portals that portals in $R$ connect to. Let $E(R,\widehat{R})$ be the edge set beween $R$ and $\widehat{R}$.
		Then, $T''$ can be modified such that the following holds.
		\begin{enumerate}
			\item The modified solution is still portal-respecting.
			\item The number of active portals for any cluster is not increased.
			\item There is at most one active portal of $C$.
			\item The resulting solution has cost increased by at most $O(w(E(R,\widehat{R})))+O(|R|) \cdot s^i$.
		\end{enumerate}
	\end{lemma}
	\begin{proof}
		Let $R$ be the active portal set of $C$.
		Observe that all portals in $\widehat{R}$ are height-$i$ or height-$(i+1)$ portals.
		
		\begin{enumerate}
			\item Pick any $u\in R$.
			\item Remove edges between $R\backslash\{u\}$ and $\widehat{R}$.
			
			\item Patching inside. Consider the subgraph $G'$ of the original $T''$ induced on the portals of $C$ and $C$'s descendant clusters. This graph may also be viewed as the ``inside'' $C$ part of $T''$, after removing edges from $R$ to $\widehat{R}$.
			Then, there exist at most $O(|R|)$ edges each of length $O(s^i)$,
			adding which makes $G'$ a connected (also Eulerian in the case of \PCTSP) graph. Include these edges to $T''$.

			\item Patching outside.
For each active portal $a\in R$, consider the set of points $R_a \subseteq \widehat{R}$ that are connected to $a$ before step 2 (where the edges are removed). Denote the removed edges between $R_a$ and $a$ as $E_a$. We add a minimum spanning tree on $R_a$ and then connect it to $u$;
observe that the edges $E_a$ together with the edge $\{a, u\}$ is a connected subgraph covering
$R_a$ and $u$.  Hence, the additional cost is at most  $O(E_a+s^i)$.

If the original graph $T$ is Eulerian, we can add some edges to make the resulting graph Eulerian as well. Notice that in either case, the additional cost for patching the crossing edges of $a$ is bounded by $O(E_a+s^i)$.
		\end{enumerate}
		
		We note that the resultant solution is still portal respecting. This is because we are only deleting and adding edges between portals of sibling clusters or those of child and parent clusters which are, by definition, still portal respecting. Hence, item 1 follows.
		Items 2 - 4 follows from the above procedure as well.  This completes the proof of Lemma~\ref{lemma:patching}.
	\end{proof}

	\begin{lemma}
		\label{lemma:reduce_crossing_cost}
		Suppose $T'$ is the solution obtained after applying the patching procedure to
		all appropriate clusters. Then,
		$E[w(T') - w(T'') \mid \mathcal{A}] \leq \epsilon\cdot w(T)$.
	\end{lemma}
	\begin{proof}
		Observe that the weight increase of $T'$ from $T''$ is due to the patching. Consider a height-$i$ cluster $C$ to which the patching is applied. Then, just before the patching, $C$ has $r' > r$ active portals.
		By Lemma~\ref{lemma:patching}, the increase of weight is at most $O(w(E(R,\widehat{R}))) + O(r') \cdot s^i$. We charge this cost to the active portals. For each active portal $a$, let $R_a\subseteq \widehat{R}$ be the portals in $R$ that is connected to $a$, and $E_a$ be the edges between $R_a$ and $a$. Then, the portal $a$ is charged with cost $O(w(E_a)) + O(s^i)$.
		
		
		We first give an upper bound for $w(E_a)$.
Observe that each node in $R_a$ is a height $i$ or $i+1$ portal and by packing property, there are at most $O(s)^k \cdot m$ such portals. Since each edge in $E_a$ is of length at most $O(s^{i+1})$, it follows that this part of the cost is at most $O(s)^k \cdot m\cdot s^{i+1}$ for all clusters with $a$ being an active portal.  The bound also dominates the second term $O(s^i)$.
		
		Hence, a height-$i$ portal takes charge at most $O(s)^k \cdot O(m) \cdot s^{i+1}$.
		
		\noindent\textbf{Charging Argument.}
		We shall ultimately charge the cost to the original solution $T$.
		Observe that a height-$i$ portal is charged only if it belongs to some cluster that is patched, so it is sufficient to charge to $T$ whenever some cluster is patched.
		
		Suppose $C$ is a cluster of height-$i$ that is patched, and $R$ is the set of active portals before patching, where $|R| > r$.
		We shall somehow charge the cost received by its portals to $T$.
		By Lemma~\ref{lemma:patching}, the patching procedure does not introduce new active portals. Hence, all active portals come from $T''$ (actually all nodes in $T''$ are active by construction).
		Let $R_j := \{ p\in R  \mid h(f(p)) = j \}$, recalling that $h(e)$ is defined to be the largest height at which some edge $e$ in $T$ is cut,
		and $f(p)$ is some edge in $T$ that caused the portal~$p$ to be added in Part~I to produce~$T''$.

		Then, $\{R_j\}_j$ is a partition of $R$. Also, it is immediate that $|R_j| = 0$ if $j < i$.
		
		\begin{lemma}
			\label{lemma:r_j_upper_bound}
			If $\mathcal{A}$ happens, then $|R_j| \leq O(1)^k \cdot q + O(\frac{s}{\epsilon})^k$, for any $j$.
		\end{lemma}
		\begin{proof}
			Let $E_j := \{ f(p) \mid p\in R_j \}$. We further partition $R_j$ into		
			$R_j^{(\text{long})} := \{ p \in R_j \mid w(f(p)) > s^j \}$ and $R_j^{(\text{short})} := \{ p \in R_j \mid w(f(p)) \leq s^j \}$.
			\ignore{
			\hubert{I changed the notation here.  Make sure everything is still consistent.}
			}
			
			\noindent\textbf{Portals activated by long edges.}
			Consider an edge $e\in E_j$ such that $w(e) > s^j$.
			Because $T$ is net-respecting, both endpoints of $e$ are in $N_{j'}$ for $s^{j'} \leq \epsilon \cdot s^j < s^{j'+1}$.
			This implies that all active portals that $e$ activates correspond to some points in $N_{j'}$.
			
			Observing that $j \geq i$, by the packing property, there are at most $O(\frac{s}{\epsilon})^k$ such portals in the height-$i$ cluster $C$.
			Therefore, $|R_j^{(\text{long})}| \leq O(\frac{s}{\epsilon})^k$.
			
			\noindent\textbf{Portals activated by short edges.}
			Consider an edge $e$ in $E_j$ such that $w(e) \leq s^j$.
			By definition, $e$ is cut at height $j$ but is not cut at height-$(j+1)$. So each such cut must be contributed by height-$j$ clusters. Notice that at least one endpoint of $e$ is within distance $2s^i$ to cluster $C$. By the definition of short edges, the other endpoint of $e$ is also within a distance of $2s^i + s^j$ from $C$. Since $j\geq i$, it follows that each cluster that cuts some short edges in $R_j^{(\text{short})}$ is within distance $4s^j$ of the center of $C$. By the packing property, there are at most $O(1)^k$ number of such clusters.
			By event $\mathcal{A}$, each such cluster contributes at most $O(q\cdot k)$ cut of edges.  Since each edge can only activate at most one portal in one cluster (in Part I), we have
			
			\begin{equation*}
			|R_j^{(\text{short})}| \leq O(1)^k \cdot O(q\cdot k) \leq O(1)^k \cdot q.
			\end{equation*}
			Combining the two cases completes the proof of Lemma~\ref{lemma:r_j_upper_bound}.

		\end{proof}
		
		Let $j' := \log_s \left( \Theta(1)^k \cdot mL \right)$.
		We charge the costs taken by $R$ to all $R_j$'s for $j \geq i + j'$ evenly.
		By Lemma~\ref{lemma:r_j_upper_bound} and $|R| > r \geq 2j' \cdot (\Theta(1)^k \cdot q + \Theta(\frac{s}{\epsilon})^k)$,
		we conclude that $|R_{\geq i + j'}| \geq \frac{|R|}{2}$.
		
		Hence, each portal in $R_{\geq i + j'}$ still takes $O(s)^k\cdot O(m)\cdot s^{i+1}$, with a slightly larger hidden constant.
		Then, we further charge this cost for each $p\in R_{\geq i + j'}$ to $f(p)$.
		
		
		\noindent\textbf{Expected Charge.}
		Finally, we use the randomness of hierarchical decomposition to bound
		the expected cost.
		
		Consider some $e$ in $T$.
		Observe that $e$ is charged only from a portal of height at most $h(e) - j'$.
		By definition, the number of portals from each height that are activated by $e$ is at most $2$.
		
		Therefore, $e$ takes cost
		\begin{align*}
			\sum_{j\leq h(e) - j'}{ O(s)^k \cdot m \cdot s^{j} } \leq  O(s)^k \cdot m \cdot s^{h(e) - j'}.
		\end{align*}
		
		\ignore{
		\hubert{I think there is no need for $L$ in the above equation.}
		}
		
		However, by definition, $e$ takes this charge only if it is cut at
		height $h(e)$, and this is with probability at most $O(k)\cdot \frac{w(e)}{s^{h(e)}}$.
		Therefore, by summing the contribution for all $L$ possible values of $h(e)$,
		the expected charge that $e$ takes conditioned on $\mathcal{A}$ is at most
		\begin{equation*}
			O(mL) \cdot O(s)^k \cdot w(e)\cdot s^{-j'} \leq \epsilon\cdot w(e).
		\end{equation*}
		
		\ignore{
		\hubert{Similarly, $L^2$ can be replaced by $L$.}
		}
		
		This completes the proof of Lemma~\ref{lemma:reduce_crossing_cost}.
	\end{proof}
	
	\noindent\textbf{Constant Probability Bound.}
	Combining Lemma~\ref{lemma:portal_resp_cost} and Lemma~\ref{lemma:reduce_crossing_cost}, we have
	$E[w(T') \mid \mathcal{A}] \leq (1+\epsilon) \cdot w(T)$.
	Observe that the optimality of $T$ implies that $w(T') \geq \frac{1}{1+\epsilon} \cdot w(T)$.
	We let $\mathcal{B}$ to be the event that $w(T')\leq w(T)$.
		
	If $\Pr[\mathcal{B} \mid \mathcal{A}] \geq 1/2$, then we are done. Otherwise, we have
\begin{align*}
\Pr[\mathcal{B} \mid \mathcal{A}](1-\epsilon)w(T)+\Pr[\mathcal{B}^c \mid \mathcal{A}]E[w(T')\mid\mathcal{A}\cap \mathcal{B}^c]\leq (1+\epsilon)w(T)
\end{align*}
It follows that
$$
E[w(T') \mid \mathcal{A}\cap \mathcal{B}^c]\leq (1+4\epsilon)w(T).
$$

By Markov's Inequality, we have
\begin{align*}
\Pr[w(T')-w(T)> 8\epsilon \cdot w(T) \mid \mathcal{A}\cap \mathcal{B}^c]\leq \frac{1}{2}.
\end{align*}

Then, we can bound the following probability as
\begin{align*}
&\quad \Pr[w(T')-w(T)> 8\epsilon \cdot w(T) \mid \mathcal{A}]\\
&=\Pr[\mathcal{B} \mid \mathcal{A}] + \Pr[\mathcal{B}^c \mid \mathcal{A}] \Pr[w(T') > (1+8\epsilon) \cdot w(T) \mid \mathcal{A} \cap \mathcal{B}^c]\\
&\leq \frac{3}{4},
\end{align*}

where the last inequality follows because $\Pr[\mathcal{B} \mid \mathcal{A}] \leq \frac{1}{2}$.
This completes the proof of Theorem~\ref{theorem:struct}.
\end{proof}

\section{A PTAS for Sparse $\PCX$ Instances: Dynamic Program}
\label{section:dp}

Suppose the estimator in Section~\ref{section:estimator}
returns a small enough value for an instance. Then,
Corollary~\ref{cor:threshold} implies that the instance has a $q$-sparse (nearly) optimal net-respecting solution,
for $q = \Theta(\frac{s k}{\epsilon})^{\Theta(k)}$.
By Theorem~\ref{theorem:struct}, to obtain $1+\epsilon$ approximation ratio,
it suffices to search for $(m,r)$-light solutions,
for some appropriate
$m := O(\frac{kL s}{\epsilon})^k$ and $r := q\cdot \Theta(1)^k + \Theta(\frac{s}{\epsilon})^k$.
We present a dynamic program algorithm to search for an optimal $(m,r)$-light solution.

Observe that Theorem~\ref{theorem:struct} assumes that some good event~$\mathcal{A}$
related to the hierarchical decomposition happens.  However, the event~$\mathcal{A}$
is defined with respect to the unknown optimal net-respecting solution.
By Proposition~\ref{prop:good_radius_constant_prop},
it is sufficient to sample a collection of $O(\log{n})$ random radii for each ball in the hierarchical decomposition.
Then, with constant probability, there is a way for each ball to choose its radius from
its collection to satisfy event~$\mathcal{A}$.

%
%
%

The dynamic program searches for $(m,r)$-solutions with respect to all possible hierarchical
decompositions obtained from the choosing the radius for each ball from the collection of sampled radii.
We show that this does not blow up the search space too much by first
describing the information needed to identify each cluster at each height.

\noindent\textbf{Information to Identify a Cluster.}
\begin{compactitem}
	\item[1.] Height $i$ and cluster center $u \in N_i$.
	This has $L \cdot O(n^k)$ combinations, recalling that
	$|N_i| \leq O(n^k)$.
	\item[2.] For each $j \geq i$, and $v \in N_j$ such
	that $d(u,v) \leq O(s^j)$, the random radius chosen by $(v,j)$.
	Observe that the space around $B(u,O(s^i))$
	can be cut by net-points in the same or higher heights that are nearby
	with respect to their distance scales.
	As argued in~\cite{DBLP:journals/siamcomp/BartalGK16},
	the number of configurations
	that are relevant to $(u,i)$ is at most $O(\log n)^{L \cdot O(1)^k} = n^{O(1)^k}$, where $L = O(\log_s n)$
	and $s = (\log n)^{\frac{c}{k}}$, for some sufficiently small universal constant $0 < c < 1$.
	\item[3.] For each $j > i$,
	which cluster at height $j$ (specified
	by the cluster center $v_j \in N_j$)
	contains
	the current cluster at height $i$.
	This has $O(1)^{kL} = n^{O(\frac{k^2}{\log \log n})}$ combinations.
\end{compactitem}

Since it is always possible to assign a direction to each edge of a tour,
our algorithm for $\PCTSP$ works on a tour in which every edge is assigned a direction.

Let $m$ and $r$ be defined as in Theorem~\ref{theorem:struct}.

\noindent\textbf{Entries of DP.}
A DP entry is identified as $(C, R, P)$, where each field is explained as follows.
\begin{itemize}
	\item $C$ denotes a cluster.
	\item $R$ denotes the subset of active portals (as defined in Section~\ref{sec:hier_decomp})
	of $C$ with $|R| \leq r$. 
	
	\item
	\begin{itemize}
		\item In $\PCTSP$, $P$ is a set of \emph{distinct}\footnote{The reason why considering distinct pairs is sufficient is explained in Lemma~\ref{lemma:unique_pair}.} ordered pairs $(p,q)$ of $R$ (where we allow $p=q$),
		such that each portal in $R$ appears in at least one pair in $P$.
		An ordered pair $(p, q)$ in $P$ means that the solution tour crosses cluster~$C$ through $(p,q)$,
		in the sense of Definition~\ref{defn:crossing_pair}.

		\item In $\PCSTP$, $P$ is a partition of $R$, where each part $U$ in $P$ corresponds to a connected component
		in the solution restricted to~$C$ that connects to portals outside the cluster via the portals in~$U$.

	\end{itemize}
\end{itemize}

\begin{lemma}[Number of Entries]
	\label{lemma:num_entry}
	The number of entries $(C, R, P)$ is at most
	$n^{O(1)^k} \cdot O(m 2^r)^r$.
\end{lemma}
\begin{proof}
	As discussed earlier, the number of cluster $C$ is at most $n^{O(1)^k}$.
	With a fixed $C$, $R$ is chosen from $m$ portals, and that $|R| \leq r$. Hence, the number of possibilities for $R$ is at most $\binom{m}{\leq r} \leq O(m)^r$.
	Finally, $P$ has at most $2^{r^2}$ possibilities, in either $\PCX$ problem.

	Therefore, the number of entries is at most $n^{O(1)^k}\cdot O(m 2^r)^r$.
\end{proof}

\noindent\textbf{Invariant: Value of an Entry.}
The value of an entry $(C, R, P)$, defined as $v(C, R, P)$, is the cost of the minimum cost graph $F$ defined on portals
of cluster~$C$ and its descendants, whose penalty is with respect to the terminals in $C$ and
 $(R,P)$ gives the connectivity requirements as follows.

	\begin{itemize}
		\item If $R = \emptyset$, then the value is the solution satisfies the same connectivity requirements
		as \PCX restricted to the sub-instance induced by~$C$.

		\item Otherwise, we have:
		\begin{itemize}
			\item For $\PCTSP$, $F$ can be partitioned into directed paths, such that
			each pair $(p,q) \in P$ corresponds to a directed path from $p$ to $q$ in $F$.
			The special case $p=q$ corresponds to a directed cycle containing $p$,
			where a degenerate cycle just containing $p$ with no edges is allowed.
			
			%
			
			\item For $\PCSTP$,  $F$ is a forest, where two portals in~$R$ are in the same connected component
			\emph{iff} they are in the same part in $P$.

		\end{itemize}
	\end{itemize}

\noindent\textbf{Evaluating a Subproblem.}
Suppose $E := (C, R, P)$ is an entry to be evaluated.
If $C$ is a height-$0$ cluster containing only one point, then
it is the base case, and its value is easily computed.

Otherwise, enumerate all possible configurations for $C$'s child clusters $I := \{ (C_i, R_i, P_i) \}_{C_i \in \mathsf{Children}(C)}$.
Then, enumerate all graphs $G$ between the portals in $R$ and $R_i$'s such that each edge
either connects
 (i) a node from $R$ to one in $R_i$, (ii) two nodes from different $R_i$'s)
or (iii) two nodes from $R$, where edges are directed for \PCTSP.  Additional edges are added among nodes within each $R_i$
to form an augmented graph $\widehat{G}$ in the following way.

\begin{itemize}
	\item For $\PCTSP$, for each~$i$ and each pair $(p,q) \in P_i$ such that $p \neq q$,
	add a directed edge $(p,q)$ to $\widehat{G}$.
	
	\ignore{
	\hubert{I don't think the ordering $\tau$ on the edges of $G$ is necessary.}
	}
	
	\ignore{		
	$G$ is a directed graph.
	Moreover, edges defined by pairs in $\bigcup_{i}{P_i}$ are added to $G$, and they are identified as \emph{virtual} edges.
	Moreover, an ordering $\tau$ on the edges of $G$ is enumerated.
	}
	
	\item For $\PCSTP$, for each $i$, for each part $U$ in $P_i$, edges from an arbitrary spanning tree on $U$
	is added to $\widehat{G}$.

\end{itemize}

The following procedure checks whether the graph~$\widehat{G}$ is \emph{consistent} with $(R,P)$.


\noindent\textbf{Consistency Checking.}
It is consistent, if all the following are true.
\begin{enumerate}
	\item If $R\neq \emptyset$, then
	for every~$i$,
	every portal in $R_i$ is connected to some portal in $R$ in $\widehat{G}$. Otherwise, all portals in $\bigcup_{i}{R_i}$ are connected in $\widehat{G}$.
	
	\item
	\begin{itemize}
		\item For $\PCTSP$.
		If $R=\emptyset$, then the directed graph~$\widehat{G}$ is Eulerian, i.e., the in-degree of every node
		equals its out-degree.
		
		Otherwise, we check that $\widehat{G}$ can be partitioned into directed paths specified by the pairs
		in $P$, 	
		where each pair $(p,q)$ corresponds to a directed path $p$ from $q$ in $\widehat{G}$.
		For $p=q$, this is a cycle containing $p$, where a degenerate cycle with an isolated $p$ is allowed.
		
		A brute force way is to consider all permutations of the edges in~$\widehat{G}$ and
		interleave the permutation  with the pairs in $P$.
		
		%
		\item For $\PCSTP$.
		If $R \neq \emptyset$, then $\widehat{G}$ is a forest
		such that two portals in $R$ are in the same connected component
		\emph{iff} they are in the same part in $P$.
		
	%
	\end{itemize}
	
\end{enumerate}

If they are consistent, then the configuration $(I, G)$ 
is a \emph{candidate} configuration for entry $E = (C, R, P)$.
The value for a candidate configuration shall be defined in the following, and
and this value is a \emph{candidate} value for $E$. The final value for $E$ is the minimum over all candidate values.

\noindent\textbf{Evaluating a Candidate Value:}
\begin{itemize}
	\item If $R=\emptyset$,
		the candidate value is
		
		$\min\{w(G) + \sum_{i : R_i \neq\emptyset}{v(C_i, R_i, P_i)} +  \sum_{i : R_i = \emptyset}{\pi(C_i)}   ,\min_{i: R_i = \emptyset}\{ v(C_i, R_i, P_i) + \pi(C \backslash C_i) \} \}$,
		where $w(G)$ is the weight of edges in $G$.
	
	\item Otherwise,
	the candidate value is $w(G) + \sum_{i : R_i \neq\emptyset}{v(C_i, R_i, P_i)} +  \sum_{i : R_i = \emptyset}{\pi(C_i)}$.
	
	%
\end{itemize}

\noindent\textbf{Final Solution.}
The final solution is corresponding to $(C_L, \emptyset, \emptyset)$, where $C_L$ is the only height-$L$ cluster.
It is easy to check that the value defined in this way satisfies the invariant. Moreover, a solution may be constructed from the values of entries.

\begin{lemma}[Running Time]
	\label{lemma:running_time}
	The time complexity of the DP is 
	$n^{O(1)^k} \cdot \exp\left(\sqrt{\log{n}} \cdot O(\frac{k}{\epsilon})^{O(k)} \right)$.
\end{lemma}
\begin{proof}
	Recall that the algorithm first enumerates an entry $E := (C, R, P)$, and then enumerate possible configurations of child entries $I := \{ (C_i, R_i, P_i) \}$, and a graph $G$ on $R$ and the $R_i$'s.
	
	
	As in Lemma~\ref{lemma:num_entry}, the number of entries $E$ is at most $n^{O(1)^k} \cdot O(mr)^r$.
	Suppose $E$ is fixed, and suppose $C$ is of height-$i$.
	We shall upper bound the number of child configurations $I$.
	Observe that there are at most $O(s)^k$ child clusters.
	As noted in~\cite{DBLP:journals/siamcomp/BartalGK16},
	the child clusters have to be consistent with $C$ on all heights at least $i$.
	Therefore, only radii on height-$(i-1)$ are not fixed, and this implies the
	number of	
	possible configurations of child clusters $\{ C_i \}_i$ is at most $ O(\log{n})^{O(1)^k \cdot O(s)^k} $.
	Hence, the number of $I$ is at most $O(\log{n})^{O(s)^k} \cdot O(m 2^r)^{r\cdot O(s)^k} \leq O(m 2^r)^{r\cdot O(s)^k}$,
	where the choice of $c$ in $s = (\log n)^{\frac{c}{k}}$ ensures that the term $O(\log{n})^{O(s)^k} = o(n)$.
	
	Since $G$ has at most $r\cdot O(s)^{k}$ vertices, and $G$ is a simple directed graph,
	the number of $G$ is at most $2^{O(r) \cdot O(s)^{k}}$.
	
	For \PCTSP, a brute force way to check consistency between $\widehat{G}$ and $(R,P)$
	takes time at most
		
	\begin{align*}
		\left(O(r^2)\cdot O(s)^{O(k)}\right)! \cdot {{O(r^2)\cdot O(s)^{O(k)}} \choose 2r}
		\leq O(rs)^{r^2 \cdot O(s)^{O(k)}}.
	\end{align*}
	
	
	In conclusion, the time complexity is
	$n^{O(1)^k} \cdot O(mr)^{r^2 \cdot O(s)^{O(k)}}$.
	
	Plugging in $m$ and $r$ defined in Theorem~\ref{theorem:struct} and the value of $q_0$, the time complexity is
	\begin{align*}
		n^{O(1)^k} \cdot O\left( \frac{skL}{\epsilon} \right)^{ O(\frac{sk}{\epsilon})^{O(k)} }
		\leq n^{O(1)^k} \cdot \exp\left(\sqrt{\log{n}} \cdot O(\frac{k}{\epsilon})^{O(k)} \right),
	\end{align*}
	where the inequality is by choosing a small enough $c$ in $s := (\log{n})^{\frac{c}{k}}$.
		
		\ignore{\hubert{How does this give the expression in Theorem~\ref{thm:main}?}}
\end{proof}

\ignore{
\noindent\textbf{$\PCTSP$: Why \emph{Distinct} Pairs Are Sufficient.}
Observe that the $(m, r)$-light solution claimed in Theorem~\ref{theorem:struct} may be changed using Lemma~\ref{lemma:unique_pair} so that the $(m, r)$-light
solution crosses each ordered pair of portals at most once.
Therefore,
although the definition of $P$ in the subproblem considers \emph{distinct} pairs only, it is still sufficient to imply the ratio.
}

\ignore{
\hubert{I think the paragraph below can be deleted.}

\noindent\textbf{Implying the Ratio.}
The $(m, r)$-light solution claimed by Theorem~\ref{theorem:struct} has an immediate correspondence to the subproblems, as follows.
Suppose the solution is $T$. Then for each cluster $C$, let $R$ be the active portals, and $P$ defined naturally.
This implies the value of the DP is no larger than $c(T)$.

Therefore, combining Lemma~\ref{lemma:running_time_tsp} and Theorem~\ref{theorem:struct} (and Lemma~\ref{lemma:unique_pair} for $\PCTSP$ case), and plugging in the value of $q_0$,
we conclude an algorithm running in polynomial time for $q_0$-sparse instances, that returns a $(1+O(\epsilon))$-approximate solution with constant probability.

}

\bibliography{ref,main}


\end{document}